\documentclass[11pt]{article}

\usepackage[margin=1in]{geometry}
\usepackage{amsmath,amssymb,amsthm,mathtools}
\usepackage{bm}
\usepackage{graphicx}
\usepackage{booktabs}
\usepackage{multirow}
\usepackage{natbib}
\usepackage[hidelinks]{hyperref}
\usepackage{enumitem}
\usepackage{float}
\floatstyle{ruled}
\newfloat{algorithm}{t}{loa}
\floatname{algorithm}{Algorithm}

\graphicspath{{./sim_out/}}
\DeclareMathOperator{\tr}{tr}
\DeclareMathOperator{\vecop}{vec}
\DeclareMathOperator{\diag}{diag}
\DeclareMathOperator{\etr}{etr}

\theoremstyle{plain}
\newtheorem{theorem}{Theorem}

\newtheorem{lemma}{Lemma}

\theoremstyle{definition}

\title{\bf Robust high-dimensional Bayesian regression with non-Gaussian errors under global--local shrinkage priors}

\author{
M. Arashi\thanks{Corresponding author: \texttt{arashi@um.ac.ir}.}\\
\small Department of Statistics, Faculty of Mathematical Sciences,\\
\small Ferdowsi University of Mashhad, Mashhad, Iran
}
\date{}

\begin{document}
\maketitle

\begin{abstract}
Multivariate regression with many correlated responses and many predictors
arises throughout the empirical sciences, yet the workhorse estimators in this
setting inherit two restrictive assumptions, i.e., that the regression coefficients
and the response dependence structure can be recovered under a Gaussian error
law, and that sparsity in the coefficient matrix can be addressed in isolation
from sparsity in the inverse error covariance. Both assumptions are routinely
violated. Macroeconomic indicators and asset returns exhibit heavy tails and
abrupt outliers; gene-expression and other high-throughput measurements are
contaminated and asymmetric. When the errors depart from normality, Gaussian
procedures lose efficiency in estimating the coefficients and, more seriously,
produce biased and non-contracting estimates of the conditional-dependence
graph. We develop a Bayesian framework for robust high-dimensional multivariate
regression that replaces the Gaussian error by a flexible scale--location
mixture --- the multivariate $t$, the skew-normal and the skew-$t$ --- and that
places a horseshoe+ global--local prior simultaneously on the regression
coefficients and on the off-diagonal entries of the error precision matrix,
thereby coupling sparsity in the regression map with sparsity in the
conditional-dependence structure of the responses. Posterior computation
proceeds either by a fully tractable Gibbs sampler, obtained from the
scale-mixture and parameter-expanded representations of the error families, or
by a mean-field variational algorithm that scales to larger problems. On the
theoretical side we establish joint posterior contraction for the coefficient
and precision matrices, selection consistency for both supports, a Kullback--
Leibler risk bound showing that the horseshoe+ prior dominates the horseshoe in
the relevant tail regime, and a bounded-sensitivity property guaranteeing that
a single arbitrarily large outlier exerts vanishing influence under the
$t$ error. A simulation study spanning four error regimes, growing sample
sizes, adversarial contamination and a range of problem dimensions confirms the
predicted behaviour, i.e., the robust estimator matches Gaussian procedures when the
errors are Gaussian and dominates them, often by wide margins, under heavy tails
and skewness. Two applications, to the FRED-MD macroeconomic database and to
daily returns of S\&P~500 constituents, illustrate the method, recovering
interpretable sparse coefficient maps and residual dependence graphs while
automatically down-weighting crisis-period observations.
\end{abstract}


\section{Introduction}
\label{sec:intro}

Let $\boldsymbol{Y}\in\mathbb{R}^{n\times q}$ collect $n$ observations on $q$
correlated response variables and let $\boldsymbol{X}\in\mathbb{R}^{n\times p}$
collect the corresponding values of $p$ predictors. The multivariate linear
regression model
\begin{eqnarray}
\boldsymbol{Y} &=& \boldsymbol{X}\boldsymbol{B} + \boldsymbol{E},
\qquad \boldsymbol{e}_i \stackrel{\mathrm{iid}}{\sim} F,
\label{eq:mvreg}
\end{eqnarray}
with coefficient matrix $\boldsymbol{B}\in\mathbb{R}^{p\times q}$ and rows
$\boldsymbol{e}_i^{\top}$ of the error matrix $\boldsymbol{E}$ drawn from a
$q$-variate law $F$ with precision matrix $\boldsymbol{\Omega}=\boldsymbol{\Sigma}^{-1}$,
is among the most heavily used tools in applied statistics. Two features of
contemporary applications strain its classical treatment.
First, the dimensions $p$ and $q$ are frequently large
relative to $n$, so that both $\boldsymbol{B}$ and $\boldsymbol{\Omega}$ must be
regularised, and the scientifically interesting objects are not the full
matrices but their \emph{supports}: which predictors act on which responses, and
which responses remain conditionally dependent once the predictors are
accounted for. Second, the Gaussian error law that underlies the dominant
methodology is, in many of the domains where \eqref{eq:mvreg} is applied,
indefensible. Returns on financial assets are leptokurtic and punctuated by
crashes; macroeconomic series contain recessions, policy shocks and, most
recently, the pandemic dislocation; high-throughput biological assays are
asymmetric and riddled with technical outliers. A handful of grossly atypical
rows of $\boldsymbol{E}$ can, under a Gaussian likelihood, distort the estimate
of $\boldsymbol{B}$ and --- because the same rows enter the residual
cross-products that identify $\boldsymbol{\Omega}$ --- corrupt the entire
conditional-dependence graph.

The literature has addressed the two halves of this problem largely separately.
On the regularisation side, penalised likelihood methods such as the MRCE
estimator of \citet{rothman2010sparse} couple an $\ell_1$ penalty on
$\boldsymbol{B}$ with a graphical-lasso \citep{friedman2008sparse,yuan2007model}
penalty on $\boldsymbol{\Omega}$, and the multivariate spike-and-slab lasso of
\citet{deshpande2019simultaneous} performs simultaneous variable and covariance
selection within a single penalised objective, with posterior contraction and
uncertainty quantification established subsequently by \citet{shen2024posterior}.
In the fully Bayesian stream, global--local shrinkage priors
\citep{carvalho2010horseshoe,polson2010shrink,bhattacharya2015dirichlet} have
proven extremely effective for sparse estimation, and \citet{bai2018high} and
\citet{zhang2019ultra} established high-dimensional posterior consistency for
multivariate regression under such priors, while the graphical horseshoe of
\citet{li2019graphical} and the precision-matrix construction of
\citet{sagar2024precision} brought global--local shrinkage to inverse-covariance
estimation. All of these developments, however, retain the Gaussian error.
On the robustness side, the Bayesian treatment of heavy tails through the
$t$ distribution \citep{lange1989robust,geweke1993bayesian}, the modelling of
skewness \citep{fernandez1998bayesian,sahu2003new,azzalini2003distributions},
and the formal theory of outlier resolution
\citep{ohagan2012bayesian,desgagne2015robustness,gagnon2020new} are mature, but
have not been integrated with high-dimensional joint shrinkage of
$\boldsymbol{B}$ and $\boldsymbol{\Omega}$.

\paragraph{Contributions.}
This paper supplies that integration and studies it both theoretically and
empirically. Our contributions are as follows.
\begin{enumerate}[label=(\roman*),leftmargin=2.2em]
\item We formulate a robust multivariate regression model in which the Gaussian
error is replaced by a scale--location mixture of normals, yielding as special
cases the multivariate $t$, the multivariate skew-normal and the multivariate
skew-$t$ law, all sharing a common precision matrix $\boldsymbol{\Omega}$
(Section~\ref{sec:model}). Each family admits an augmentation that renders the
conditional posteriors conjugate.
\item We place a horseshoe+ prior \citep{bhadra2017horseshoeplus}
simultaneously on the entries of $\boldsymbol{B}$ and on the off-diagonal
entries of $\boldsymbol{\Omega}$, so that the regression map and the
conditional-dependence graph are shrunk jointly. The horseshoe+ is chosen
deliberately, i.e., its heavier tail leaves genuine signals essentially unbiased,
which we show translates into a tighter risk bound and into the robustness
property below.
\item We derive a Gibbs sampler whose every full conditional is available in
closed form (Section~\ref{sec:computation}), exploiting the
Makalic--Schmidt \citep{makalic2016simple} representation of the horseshoe+ and a
column-wise block update for the graphical horseshoe+ in the spirit of
\citet{li2019graphical}. We complement it with a mean-field variational
algorithm that delivers one-to-two orders of magnitude in speed-up.
\item We establish joint posterior contraction of $\boldsymbol{B}$ and $\boldsymbol{\Omega}$ at
near-minimax rates; consistency of the induced variable and edge-selection
rules; a Kullback--Leibler risk bound showing that the horseshoe+ dominates the
horseshoe in the sparse-signal regime; and a bounded-sensitivity theorem showing
that under the $t$ error a single unbounded outlier has asymptotically no effect
on the posterior mean of $\boldsymbol{B}$.
\end{enumerate}

The remainder of the paper is organised as follows.
Section~\ref{sec:prelim} fixes notation and recalls the scale-mixture
representations on which the construction rests. Section~\ref{sec:model}
introduces the model and the joint horseshoe+ prior. Section~\ref{sec:computation}
develops the Gibbs sampler and the variational algorithm. Section~\ref{sec:theory}
contains the theory; all proofs are deferred to the appendix.
Sections~\ref{sec:sim} and \ref{sec:real} report the simulation study and the
applications, and Section~\ref{sec:discussion} concludes. Proofs and additional
derivations are collected in the appendix.

\section{Preliminaries and notation}
\label{sec:prelim}

For a matrix $\boldsymbol{A}$ we write $\boldsymbol{A}^{\top}$ for its transpose,
$\tr(\boldsymbol{A})$ for its trace, $\vecop(\boldsymbol{A})$ for the
column-stacking vectorisation, $|\boldsymbol{A}|$ for the determinant and
$\etr(\boldsymbol{A})=\exp\{\tr(\boldsymbol{A})\}$. The Frobenius norm is
$\|\boldsymbol{A}\|_F=\{\tr(\boldsymbol{A}^{\top}\boldsymbol{A})\}^{1/2}$, the
spectral norm is $\|\boldsymbol{A}\|_2$, and $\otimes$ denotes the Kronecker
product, for which $\vecop(\boldsymbol{A}\boldsymbol{C}\boldsymbol{D})
=(\boldsymbol{D}^{\top}\otimes\boldsymbol{A})\vecop(\boldsymbol{C})$. We write
$\boldsymbol{a}_i^{\top}$ for the $i$th row and $\boldsymbol{a}_{\cdot k}$ for the
$k$th column of a matrix. For symmetric positive-definite
$\boldsymbol{\Sigma}\in\mathbb{R}^{q\times q}$, the matrix-normal density of
$\boldsymbol{E}$ with independent rows of precision
$\boldsymbol{\Omega}=\boldsymbol{\Sigma}^{-1}$ is proportional to
$|\boldsymbol{\Omega}|^{n/2}\etr(-\tfrac12\boldsymbol{E}\boldsymbol{\Omega}
\boldsymbol{E}^{\top})$. The half-Cauchy density on $(0,\infty)$ with scale $b$
is $\mathcal{C}^{+}(x\mid b)=2/[\pi b\{1+(x/b)^2\}]$, and we write
$\mathrm{IG}(a,b)$ for the inverse-gamma and $\mathrm{Ga}(a,b)$ for the gamma law
with shape $a$ and rate $b$.

\subsection{Scale--location mixtures of normals}
\label{sec:mixtures}

The error families we employ all arise as mixtures of a Gaussian over a positive
mixing variable, possibly shifted by a latent truncated-normal term. This is the
structural feature that makes posterior computation tractable. Three cases are
of primary interest.

\emph{Multivariate $t$.} The $q$-variate Student-$t$ law with $\nu$ degrees of
freedom, location $\boldsymbol{0}$ and precision $\boldsymbol{\Omega}$ has the
representation
\begin{eqnarray}
\boldsymbol{e}_i \mid w_i &\sim&
\mathcal{N}_q\!\left(\boldsymbol{0},\, w_i^{-1}\boldsymbol{\Omega}^{-1}\right),
\qquad
w_i \sim \mathrm{Ga}\!\left(\tfrac{\nu}{2},\,\tfrac{\nu}{2}\right),
\label{eq:t-mix}
\end{eqnarray}
so that integrating out the observation-specific weight $w_i$ returns the
$t$ density \citep{lange1989robust,geweke1993bayesian}. Small values of $w_i$
correspond to inflated variance and hence to observations that the model treats
as outliers; the posterior mean of $w_i$ is therefore a natural diagnostic, used
repeatedly in Section~\ref{sec:real}.

\emph{Multivariate skew-normal and skew-$t$.} Following the construction of
\citet{sahu2003new}, asymmetry is introduced through a latent half-normal vector
$\boldsymbol{t}_i$. Writing $\boldsymbol{\delta}=(\delta_1,\dots,\delta_q)^{\top}$
for the skewness parameters and $\boldsymbol{\Delta}=\diag(\boldsymbol{\delta})$,
\begin{eqnarray}
\boldsymbol{e}_i \mid \boldsymbol{t}_i, w_i &\sim&
\mathcal{N}_q\!\left(\boldsymbol{\Delta}\boldsymbol{t}_i,\,
w_i^{-1}\boldsymbol{\Omega}^{-1}\right),
\qquad
\boldsymbol{t}_i \mid w_i \sim
\mathcal{N}_q^{+}\!\left(\boldsymbol{0},\, w_i^{-1}\boldsymbol{I}_q\right),
\label{eq:st-mix}
\end{eqnarray}
where $\mathcal{N}_q^{+}$ denotes the $q$-variate normal truncated to the
positive orthant. Taking $w_i\equiv1$ yields the skew-normal family; taking
$w_i\sim\mathrm{Ga}(\nu/2,\nu/2)$ as in \eqref{eq:t-mix} yields the skew-$t$
family, which simultaneously accommodates heavy tails and asymmetry
\citep{azzalini2003distributions}. The Gaussian model is recovered as
$\boldsymbol{\delta}=\boldsymbol{0}$, $\nu\to\infty$.

A unifying way to read \eqref{eq:t-mix}--\eqref{eq:st-mix} is that, conditional
on the latent quantities $(w_i,\boldsymbol{t}_i)$, model~\eqref{eq:mvreg} is an
ordinary Gaussian multivariate regression with a heteroscedastic row weight and
a row-specific mean shift. Every step of the sampler in
Section~\ref{sec:computation} operates on this conditionally Gaussian layer.

\section{Model and joint horseshoe+ prior}
\label{sec:model}

\subsection{The robust likelihood}

Combining \eqref{eq:mvreg} with the mixture representation, the working
likelihood for the skew-$t$ model is, conditional on the latent variables,
\begin{eqnarray}
\boldsymbol{y}_i \mid \boldsymbol{B},\boldsymbol{\Omega},\boldsymbol{\delta},
w_i,\boldsymbol{t}_i
&\sim&
\mathcal{N}_q\!\left(\boldsymbol{B}^{\top}\boldsymbol{x}_i
+ \boldsymbol{\Delta}\boldsymbol{t}_i,\; w_i^{-1}\boldsymbol{\Omega}^{-1}\right),
\qquad i=1,\dots,n,
\label{eq:lik}
\end{eqnarray}
with the $t$ model obtained by setting $\boldsymbol{\delta}=\boldsymbol{0}$ and
the Gaussian model by additionally fixing $w_i\equiv1$. The degrees of freedom
$\nu$ may be fixed at a small value, providing strong protection against
outliers, or assigned a prior and updated; throughout the paper we fix
$\nu\in\{3,4\}$ to expose the robustness behaviour cleanly, and discuss the
estimation of $\nu$ in Section~\ref{sec:discussion}.

\subsection{Joint horseshoe+ shrinkage}
\label{sec:prior}

The novelty of the prior is that the same global--local mechanism is applied to
two distinct sparse objects. For the regression coefficients we adopt a
horseshoe+ prior \citep{bhadra2017horseshoeplus} on each entry,
\begin{eqnarray}
\beta_{jk}\mid \lambda_{jk},\eta_{jk},\tau
&\sim& \mathcal{N}\!\left(0,\, \tau^2 \lambda_{jk}^2 \eta_{jk}^2\right),
\quad
\lambda_{jk}\sim\mathcal{C}^{+}(0,1),
\quad
\eta_{jk}\sim\mathcal{C}^{+}(0,1),
\label{eq:hsplus-B}
\end{eqnarray}
for $j=1,\dots,p$, $k=1,\dots,q$, with a global scale $\tau\sim\mathcal{C}^{+}(0,1)$.
The horseshoe+ augments the horseshoe with the second half-Cauchy factor
$\eta_{jk}$; the resulting marginal prior on $\beta_{jk}$ has an even taller
spike at the origin and a heavier tail than the horseshoe, the two features that
respectively encourage aggressive shrinkage of noise and near-unbiased treatment
of signal. Figure~\ref{fig:priors}(a) contrasts the marginal densities, and
Figure~\ref{fig:priors}(b) the implied posterior-mean shrinkage profiles, against
the Laplace prior that underlies the lasso: the horseshoe and horseshoe+ track
the $45^{\circ}$ line for large observations, whereas the Laplace imposes a
constant downward bias that never abates.

\begin{figure}[t]
\centering
\includegraphics[width=\textwidth]{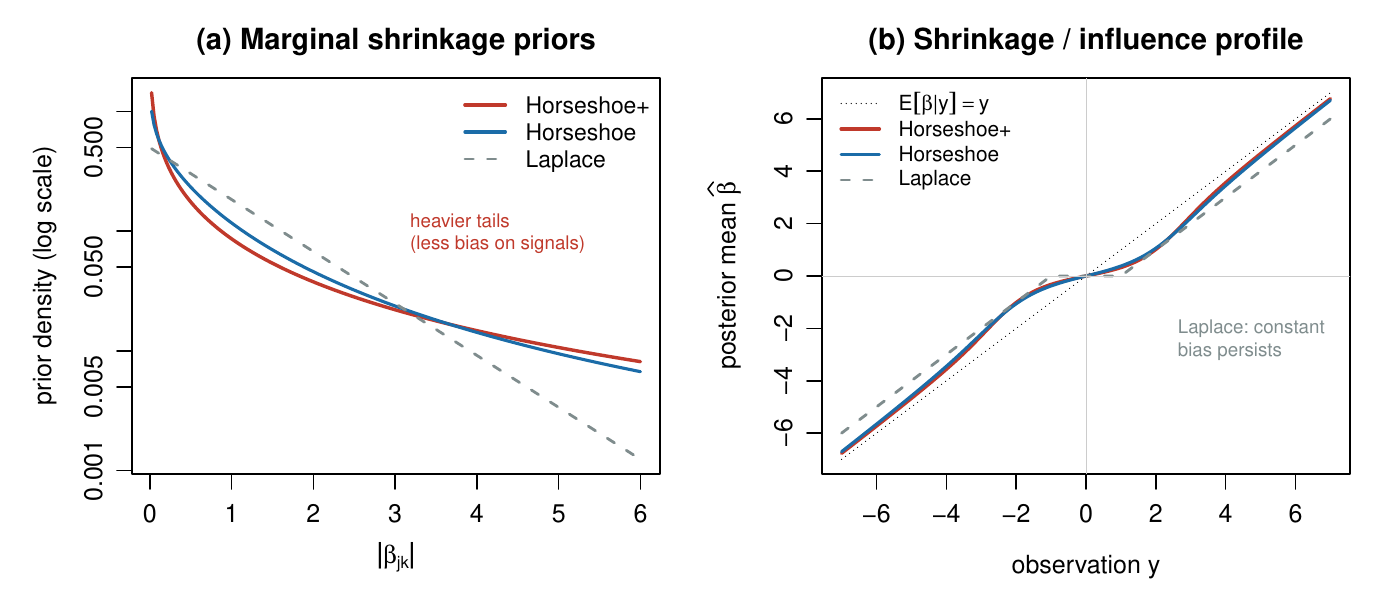}
\caption{(a) Marginal prior densities of a coefficient under the Laplace,
horseshoe and horseshoe+ priors on a logarithmic scale; the horseshoe+ has the
heaviest tail. (b) Implied posterior-mean shrinkage profile $E[\beta\mid y]$ in
the unit-variance normal-means problem. The horseshoe and horseshoe+ leave large
signals essentially unbiased (the curve approaches $E[\beta\mid y]=y$), whereas
the Laplace/lasso retains a constant bias. This contrast underlies both the risk
bound of Theorem~\ref{thm:kl} and the robustness property of
Theorem~\ref{thm:influence}.}
\label{fig:priors}
\end{figure}

For the precision matrix we place a horseshoe+ prior on the off-diagonal entries
and a diffuse exponential prior on the diagonal, following the graphical
horseshoe paradigm of \citet{li2019graphical} but with the horseshoe+ local
structure:
\begin{eqnarray}
\omega_{kl}\mid \gamma_{kl},\zeta_{kl},\rho
&\sim& \mathcal{N}\!\left(0,\, \rho^2 \gamma_{kl}^2 \zeta_{kl}^2\right),
\quad k<l, \qquad
\omega_{kk}\sim\mathrm{Exp}(\theta/2),
\label{eq:hsplus-Omega}
\end{eqnarray}
with $\gamma_{kl},\zeta_{kl}\sim\mathcal{C}^{+}(0,1)$, global scale
$\rho\sim\mathcal{C}^{+}(0,1)$, a fixed rate hyperparameter $\theta>0$ for the
diagonal entries, and the prior restricted to the cone of positive
definite matrices. The construction shrinks weak conditional dependencies to
zero while leaving strong ones intact, and (Section~\ref{sec:real}) recovers
graphs that are both sparse and interpretable.

\begin{figure}[t]
\centering
\includegraphics[width=\textwidth]{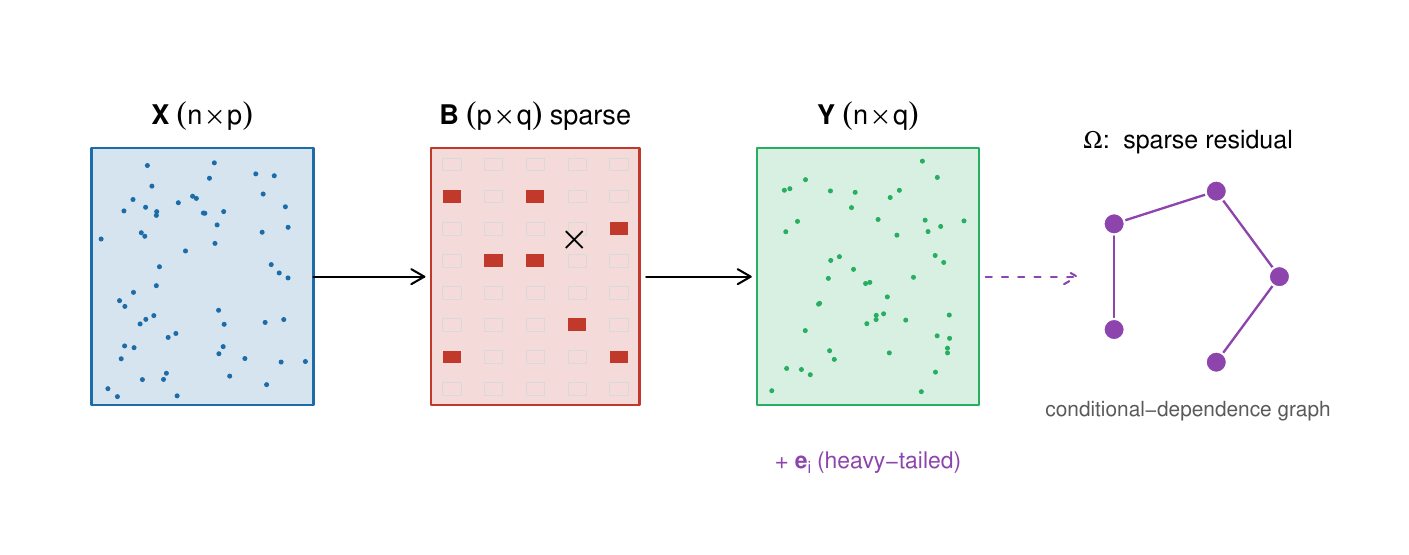}
\caption{Schematic of the model. Predictors $\boldsymbol{X}$ act through a sparse
coefficient matrix $\boldsymbol{B}$ on the responses $\boldsymbol{Y}$; the errors
$\boldsymbol{e}_i$ follow a heavy-tailed scale--location mixture; and the residual
precision $\boldsymbol{\Omega}$ encodes a sparse conditional-dependence graph
among the responses. The horseshoe+ prior is applied jointly to $\boldsymbol{B}$
and to the off-diagonal entries of $\boldsymbol{\Omega}$.}
\label{fig:schematic}
\end{figure}

The local scales $\lambda_{jk},\eta_{jk}$ and $\gamma_{kl},\zeta_{kl}$ adapt to
each coordinate, while the global scales $\tau$ and $\rho$ control overall
sparsity; \citet{piironen2017sparsity} discuss the role of such global scales and
their relationship to the effective number of nonzero coefficients.

\section{Posterior computation}
\label{sec:computation}

\subsection{Gibbs sampler}
\label{sec:gibbs}

We use the standard half-Cauchy-as-inverse-gamma-mixture identity
\citep{makalic2016simple}: if $x^2\mid a\sim\mathrm{IG}(1/2,1/a)$ and
$a\sim\mathrm{IG}(1/2,1)$ then $x\sim\mathcal{C}^{+}(0,1)$. Introducing
auxiliary variables for every half-Cauchy factor turns
\eqref{eq:hsplus-B}--\eqref{eq:hsplus-Omega} into a conditionally Gaussian--
inverse-gamma hierarchy, so that all updates are available in closed form. We
summarise the sampler; the full conditional derivations are given in
Appendix~\ref{app:gibbs}.

\emph{Coefficient matrix.} Conditional on the latent weights, the skewness term
and $\boldsymbol{\Omega}$, write $\boldsymbol{W}=\diag(w_1,\dots,w_n)$ and let
$\boldsymbol{R}=\boldsymbol{Y}-\boldsymbol{T}\boldsymbol{\Delta}$ be the residual
after removing the skewness shift, where $\boldsymbol{T}$ stacks the
$\boldsymbol{t}_i^{\top}$. With $\boldsymbol{\Lambda}_{\boldsymbol{B}}
=\diag\{(\tau^2\lambda_{jk}^2\eta_{jk}^2)^{-1}\}$ the prior precision of
$\vecop(\boldsymbol{B})$, the full conditional is Gaussian,
\begin{eqnarray}
\vecop(\boldsymbol{B})\mid\cdot
&\sim& \mathcal{N}\!\left(\boldsymbol{m}_{\boldsymbol{B}},
\boldsymbol{V}_{\boldsymbol{B}}\right),
\qquad
\boldsymbol{V}_{\boldsymbol{B}}^{-1}
= \boldsymbol{\Omega}\otimes(\boldsymbol{X}^{\top}\boldsymbol{W}\boldsymbol{X})
+ \boldsymbol{\Lambda}_{\boldsymbol{B}},
\label{eq:B-update}
\cr
\boldsymbol{m}_{\boldsymbol{B}}
&=& \boldsymbol{V}_{\boldsymbol{B}}\,
\vecop\!\left(\boldsymbol{X}^{\top}\boldsymbol{W}\boldsymbol{R}\,\boldsymbol{\Omega}\right).
\nonumber
\end{eqnarray}

\emph{Precision matrix.} The off-diagonal horseshoe+ prior makes the graphical
update a sequence of $q$ column-wise Gaussian draws. Let
$\boldsymbol{E}=\boldsymbol{R}-\boldsymbol{X}\boldsymbol{B}
=\boldsymbol{Y}-\boldsymbol{X}\boldsymbol{B}-\boldsymbol{T}\boldsymbol{\Delta}$
denote the full residual and
$\boldsymbol{S}=\boldsymbol{E}^{\top}\boldsymbol{W}\boldsymbol{E}$ its weighted
cross-product. Partitioning $\boldsymbol{\Omega}$ and $\boldsymbol{S}$ with respect
to the $k$th coordinate --- so that $\boldsymbol{\Omega}_{11}$ is the leading
$(q-1)\times(q-1)$ block, $\boldsymbol{\omega}_{12}$ the $(q-1)$-vector of
off-diagonal entries of column $k$ and $\omega_{22}$ the corresponding diagonal
entry, with $\boldsymbol{s}_{12},s_{22}$ the matching blocks of $\boldsymbol{S}$
--- the column is updated by the block scheme of \citet{li2019graphical}, namely
\begin{eqnarray}
\boldsymbol{\omega}_{12}\mid\cdot
&\sim& \mathcal{N}\!\left(-\boldsymbol{C}^{-1}\boldsymbol{s}_{12},\,
\boldsymbol{C}^{-1}\right),
\qquad
\boldsymbol{C} = s_{22}\,\boldsymbol{\Omega}_{11}^{-1}
+ \diag\{(\rho^2\gamma_{kl}^2\zeta_{kl}^2)^{-1}\},
\label{eq:Omega-update}
\cr
\gamma_{k}\mid\cdot &\sim& \mathrm{Ga}\!\left(\tfrac{n}{2}+1,\,
\tfrac{s_{22}}{2}\right),
\qquad
\omega_{22}=\gamma_{k}+\boldsymbol{\omega}_{12}^{\top}
\boldsymbol{\Omega}_{11}^{-1}\boldsymbol{\omega}_{12},
\nonumber
\end{eqnarray}
where $\gamma_{k}>0$ is the Schur complement of the $k$th coordinate (not to be
confused with the local scale $\gamma_{kl}$ in $\boldsymbol{C}$); writing
$\omega_{22}=\gamma_{k}+\boldsymbol{\omega}_{12}^{\top}\boldsymbol{\Omega}_{11}^{-1}
\boldsymbol{\omega}_{12}$ guarantees that the updated $\boldsymbol{\Omega}$ remains
positive definite. The local and global scales follow from the inverse-gamma
augmentation.

\emph{Latent weights and skewness.} The mixing weights have conditionally
independent gamma posteriors,
\begin{eqnarray}
w_i\mid\cdot &\sim&
\mathrm{Ga}\!\left(\tfrac{\nu+q+\mathbb{1}\{\text{skew}\}\,q}{2},\;
\tfrac{\nu + d_i}{2}\right),
\qquad
d_i = (\boldsymbol{y}_i-\boldsymbol{B}^{\top}\boldsymbol{x}_i
-\boldsymbol{\Delta}\boldsymbol{t}_i)^{\top}\boldsymbol{\Omega}
(\boldsymbol{y}_i-\boldsymbol{B}^{\top}\boldsymbol{x}_i
-\boldsymbol{\Delta}\boldsymbol{t}_i) + s_i,
\label{eq:w-update}
\end{eqnarray}
where $s_i=\boldsymbol{t}_i^{\top}\boldsymbol{t}_i$ in the skew-$t$ case and
$s_i=0$ otherwise. Each latent half-normal coordinate $t_{ik}$ is drawn from a
univariate truncated normal, and the skewness parameters $\boldsymbol{\delta}$
have a Gaussian full conditional. The complete algorithm is collected as
Algorithm~\ref{alg:gibbs} in Appendix~\ref{app:gibbs}.

\subsection{Variational inference}
\label{sec:vb}

For settings in which the Gibbs sampler is too costly, we derive a mean-field
variational approximation
$\boldsymbol{q}(\boldsymbol{B},\boldsymbol{\Omega},\boldsymbol{\delta},
\{w_i,\boldsymbol{t}_i\},\cdot)$ that factorises across blocks. The
coordinate-ascent updates mirror the Gibbs full conditionals with expectations
replacing draws; in particular the weights enter only through their means
$E[w_i]=(\nu+q)/(\nu+E[d_i])$, which again identify down-weighted observations.
The algorithm is given in Appendix~\ref{app:vb}. Section~\ref{sec:sim} shows that
it attains one-to-two orders of magnitude in speed-up at accuracy comparable to
the sampler in well-identified regimes.

\section{Theoretical results}
\label{sec:theory}

We study the model in the regime where $p=p_n$ and $q=q_n$ may grow with $n$.
Let $\boldsymbol{B}_0$ and $\boldsymbol{\Omega}_0$ denote the true coefficient and
precision matrices, with row-support size $s_{\boldsymbol{B}}=
\#\{(j,k):\beta_{0,jk}\neq0\}$ and edge-support size
$s_{\boldsymbol{\Omega}}=\#\{(k,l):k<l,\ \omega_{0,kl}\neq0\}$. We impose the
following conditions.
\begin{itemize}[leftmargin=2.4em]
\item[(A1)] \emph{Dimension growth.} $\log(p_n q_n)=o(n)$ and
$(s_{\boldsymbol{B}}+s_{\boldsymbol{\Omega}})\log(p_n q_n)=o(n)$.
\item[(A2)] \emph{Design regularity.} The eigenvalues of
$n^{-1}\boldsymbol{X}^{\top}\boldsymbol{X}$ are bounded away from $0$ and
$\infty$ uniformly in $n$; equivalently there exist
$0<c_1\le c_2<\infty$ with $c_1\le\Lambda_{\min}(n^{-1}\boldsymbol{X}^{\top}
\boldsymbol{X})\le\Lambda_{\max}(n^{-1}\boldsymbol{X}^{\top}\boldsymbol{X})\le c_2$.
\item[(A3)] \emph{Spectrum of the precision.} There exist constants
$0<\underline{k}\le\overline{k}<\infty$ such that
$\underline{k}\le\Lambda_{\min}(\boldsymbol{\Omega}_0)\le
\Lambda_{\max}(\boldsymbol{\Omega}_0)\le\overline{k}$.
\item[(A4)] \emph{Mixing law.} The mixing variable $w$ has a density supported on
$(0,\infty)$ with $E[w]<\infty$ and $E[w^{-1}]<\infty$, and is log-regularly
varying at the origin in the sense of \citet{desgagne2015robustness}; the
$t$ mixing law of \eqref{eq:t-mix} satisfies this for every $\nu>0$.
\item[(A5)] \emph{Signal strength.} The nonzero entries of $\boldsymbol{B}_0$ and
$\boldsymbol{\Omega}_0$ exceed a detection threshold of order
$\sqrt{\log(p_nq_n)/n}$, uniformly in $n$.
\end{itemize}

Conditions (A1)--(A3) and (A5) are standard in the high-dimensional Bayesian
literature \citep{castillo2015bayesian,bai2018high,li2019graphical}. Condition
(A4) is the robustness hypothesis: it is what allows a single outlier to be
absorbed by the mixing weight rather than transmitted to the estimates.

\subsection{Joint posterior contraction}

\begin{theorem}[Joint contraction]
\label{thm:contraction}
Under \textnormal{(A1)--(A5)} and the joint horseshoe+ prior of
Section~\ref{sec:prior}, there is a constant $M>0$ such that, with
$\epsilon_n^2 = (s_{\boldsymbol{B}}+s_{\boldsymbol{\Omega}}+q_n)
\log(p_nq_n)/n$,
\begin{eqnarray}
\Pi\!\left(\|\boldsymbol{B}-\boldsymbol{B}_0\|_F
+ \|\boldsymbol{\Omega}-\boldsymbol{\Omega}_0\|_F > M\epsilon_n
\,\middle|\, \boldsymbol{Y},\boldsymbol{X}\right)
&\longrightarrow& 0
\label{eq:contraction}
\end{eqnarray}
in probability as $n\to\infty$.
\end{theorem}

The rate $\epsilon_n$ matches, up to constants, the rates obtained for Gaussian
multivariate regression under global--local priors \citep{bai2018high,
zhang2019ultra} and for the graphical horseshoe \citep{li2019graphical}; the
contribution of Theorem~\ref{thm:contraction} is that the rate is retained
\emph{jointly} and \emph{under the non-Gaussian error}. The proof, in
Appendix~\ref{app:contraction}, verifies the prior-mass and testing conditions of
the general contraction theorem of \citet{ghosal2000convergence} for the
augmented model, controlling the mixing weights through (A4).

\subsection{Selection consistency}

\begin{theorem}[Selection consistency]
\label{thm:selection}
Under the conditions of Theorem~\ref{thm:contraction}, let
$\widehat{S}_{\boldsymbol{B}}$ and $\widehat{S}_{\boldsymbol{\Omega}}$ be the
supports selected by thresholding the marginal posterior credible intervals at
the $50\%$ level. Then
\begin{eqnarray}
\Pi\!\left(\widehat{S}_{\boldsymbol{B}}=S_{\boldsymbol{B}},\;
\widehat{S}_{\boldsymbol{\Omega}}=S_{\boldsymbol{\Omega}}
\,\middle|\,\boldsymbol{Y},\boldsymbol{X}\right)\longrightarrow 1
\end{eqnarray}
in probability as $n\to\infty$.
\end{theorem}

\subsection{Risk advantage of the horseshoe+}

The next result formalises the intuition of Figure~\ref{fig:priors}. Consider the
canonical sparse normal-means problem embedded in a single coordinate of
$\boldsymbol{B}$, $y=\beta_0+\varepsilon$ with $\varepsilon\sim\mathcal{N}(0,1)$,
and let $\widehat{\beta}_{\mathrm{HS}}$ and $\widehat{\beta}_{\mathrm{HS}^{+}}$
be the posterior means under the horseshoe and horseshoe+ priors.

\begin{theorem}[Risk dominance of the horseshoe+ on signals]
\label{thm:kl}
Consider the normal-means problem with horseshoe and horseshoe+ priors under a
common global scale, and let
\begin{eqnarray*}
R(\widehat{\beta})=E_{\beta_0}\,\mathrm{KL}\{\mathcal{N}(\beta_0,1)\,\|\,
\mathcal{N}(\widehat{\beta},1)\}=\tfrac12\,E_{\beta_0}(\widehat{\beta}-\beta_0)^2	
\end{eqnarray*}
denote the Kullback--Leibler risk, the expectation being over
$y\sim\mathcal{N}(\beta_0,1)$. The posterior-mean biases on a fixed signal satisfy,
as $|y|\to\infty$,
\begin{eqnarray}
\bigl|\,E[\beta\mid y]-y\,\bigr|_{\mathrm{HS}}
&=& \Theta\!\left(\frac{1}{|y|}\right),
\qquad
\bigl|\,E[\beta\mid y]-y\,\bigr|_{\mathrm{HS}^{+}}
\;=\; \Theta\!\left(\frac{1}{|y|\,\log|y|}\right),
\label{eq:biasrates}
\end{eqnarray}
so that the horseshoe+ bias is smaller than the horseshoe bias by a factor of
order $\log|y|$. Consequently there is a threshold $b_0<\infty$ such that, for
every true signal with $|\beta_0|\ge b_0$,
\begin{eqnarray}
R(\widehat{\beta}_{\mathrm{HS}^{+}})
&<& R(\widehat{\beta}_{\mathrm{HS}}).
\end{eqnarray}
\end{theorem}

The improvement is driven by the additional logarithmic factor in the tail of the
horseshoe+ marginal \citep{bhadra2017horseshoeplus,datta2013asymptotic}, which
reduces the bias incurred on large signals; the proof is in
Appendix~\ref{app:kl}.

\subsection{Bounded sensitivity to outliers}

Finally we make precise the sense in which the $t$ error confers robustness. Fix
all observations but the first, and let
$\boldsymbol{y}_1=\boldsymbol{y}_1(a)$ be displaced along a fixed direction
$\boldsymbol{u}$ with $\|\boldsymbol{u}\|=1$ by an amount $a$, i.e.\
$\boldsymbol{y}_1(a)=\boldsymbol{y}_1^{\circ}+a\boldsymbol{u}$. Denote by
$\widehat{\boldsymbol{B}}(a)=E[\boldsymbol{B}\mid\boldsymbol{Y}(a),\boldsymbol{X}]$
the posterior mean as a function of the contamination $a$.

\begin{theorem}[Bounded sensitivity]
\label{thm:influence}
Under the $t$ error \eqref{eq:t-mix} with any $\nu>0$ and the prior of
Section~\ref{sec:prior},
\begin{eqnarray}
\lim_{|a|\to\infty}\;
\left\|\,\widehat{\boldsymbol{B}}(a)-\widehat{\boldsymbol{B}}(\pm\infty)\,\right\|_F
&=& 0,
\end{eqnarray}
that is, the posterior mean of $\boldsymbol{B}$ converges to a finite limit as
the outlier moves to infinity, and the gross-error sensitivity is finite. Under
the Gaussian error the same quantity diverges linearly in $|a|$.
\end{theorem}

The proof, in Appendix~\ref{app:influence}, shows that as $|a|\to\infty$ the
posterior of the weight $w_1$ concentrates near $0$, so that the offending row is
asymptotically excised from the likelihood; this is the whole-robustness
phenomenon of \citet{desgagne2015robustness} and \citet{gagnon2020new} realised
inside the multivariate regression. Section~\ref{sec:sim} confirms the linear
divergence of the Gaussian estimator and the bounded behaviour of the robust one.

\section{Simulation study}
\label{sec:sim}

We evaluate the method (\textsc{RHS-MR}) against three Gaussian-likelihood
benchmarks: the same horseshoe+ model fitted under a Gaussian error
(\textsc{RHS-MR (Gaussian)}, an ablation isolating the effect of the robust
likelihood); a two-step estimator that fits response-wise lasso regressions and
a graphical lasso on the residuals (\textsc{Sep-Lasso+GLasso}), in the spirit of
the MRCE estimator of \citet{rothman2010sparse}; and ordinary least squares with
the sample residual precision (\textsc{OLS+Sample}). Performance is measured by
the Frobenius estimation errors
$\|\widehat{\boldsymbol{B}}-\boldsymbol{B}\|_F$ and
$\|\widehat{\boldsymbol{\Omega}}-\boldsymbol{\Omega}\|_F$, by the Matthews
correlation coefficient (MCC) of the recovered supports, and by the predictive
mean-squared error (PMSE) on an independent test sample. All reported numbers are
averages over independent replications; the data-generating code and the seeds
are released with the paper. Throughout, the true precision $\boldsymbol{\Omega}$
has an autoregressive (tridiagonal) structure and $\boldsymbol{B}$ has a sparse
row support.

\subsection{Experiment S1: four error regimes}

Table~\ref{tab:main} reports results for $n=100$, $p=30$, $q=6$ over $15$
replications under four error laws: Gaussian, $t_3$, skew-normal and skew-$t_4$.
For the non-Gaussian families the robust fit uses the multivariate-$t$
likelihood. Three patterns emerge. Under Gaussian errors the robust and Gaussian
versions of \textsc{RHS-MR} are indistinguishable (coefficient error $0.58$
versus $0.57$), so robustness costs essentially nothing when it is not needed.
Under $t_3$ errors the robust fit cuts the coefficient error from $1.03$ to
$0.65$ --- a $37\%$ reduction relative to the Gaussian ablation --- and improves
the precision-matrix error from $2.96$ to $2.01$ and its MCC from $0.59$ to
$0.80$. Under skew-$t_4$ errors, which combine heavy tails with asymmetry, the
gain is largest: the coefficient error falls from $2.40$ to $1.37$, a $43\%$
reduction. Across all four regimes \textsc{RHS-MR} dominates the penalised and
unpenalised benchmarks, often halving their errors.

\begin{table}[t]
\centering
\caption{Experiment S1. Mean estimation error, support recovery and prediction
over $15$ replications ($n=100$, $p=30$, $q=6$). For the non-Gaussian families
the robust model uses the multivariate-$t$ likelihood; \textsc{RHS-MR
(Gaussian)} is the same prior under a Gaussian error.}
\label{tab:main}
\small
% main comparison table
\begin{tabular}{llccccc}
\hline\hline
Errors & Method & $\|\widehat{\boldsymbol{B}}-\boldsymbol{B}\|_F$ & $\mathrm{MCC}_{\boldsymbol{B}}$ & $\|\widehat{\boldsymbol{\Omega}}-\boldsymbol{\Omega}\|_F$ & $\mathrm{MCC}_{\boldsymbol{\Omega}}$ & PMSE \\
\hline
\multirow{4}{*}{Gaussian} & RHS-MR (robust) & 0.58 & 0.99 & 1.35 & 0.77 & 6.36 \\
 & RHS-MR (Gaussian) & 0.57 & 0.99 & 1.32 & 0.79 & 6.36 \\
 & Sep-Lasso+GLasso & 0.95 & 0.52 & 1.75 & 0.75 & 6.96 \\
 & OLS+Sample & 1.61 & 0.50 & 4.48 & 0.12 & 8.50 \\
\hline
\multirow{4}{*}{$t_3$} & RHS-MR (robust) & 0.64 & 1.00 & 2.01 & 0.80 & 17.22 \\
 & RHS-MR (Gaussian) & 1.03 & 0.95 & 2.96 & 0.59 & 17.98 \\
 & Sep-Lasso+GLasso & 1.54 & 0.52 & 3.55 & 0.60 & 19.34 \\
 & OLS+Sample & 2.55 & 0.37 & 2.33 & 0.19 & 23.68 \\
\hline
\multirow{4}{*}{Skew-normal} & RHS-MR (robust) & 1.14 & 0.95 & 3.01 & 0.39 & 21.12 \\
 & RHS-MR (Gaussian) & 1.07 & 0.95 & 3.84 & 0.40 & 20.97 \\
 & Sep-Lasso+GLasso & 1.51 & 0.75 & 4.21 & 0.53 & 22.02 \\
 & OLS+Sample & 2.22 & 0.40 & 3.21 & 0.26 & 24.89 \\
\hline
\multirow{4}{*}{Skew-$t_4$} & RHS-MR (robust) & 1.37 & 0.87 & 3.54 & 0.47 & 56.00 \\
 & RHS-MR (Gaussian) & 2.40 & 0.65 & 4.76 & 0.32 & 60.16 \\
 & Sep-Lasso+GLasso & 2.46 & 0.66 & 4.93 & 0.32 & 60.28 \\
 & OLS+Sample & 3.89 & 0.30 & 4.34 & 0.28 & 69.17 \\
\hline
\hline
\end{tabular}

\end{table}

\subsection{Experiment S2: posterior contraction}

To probe Theorem~\ref{thm:contraction} we fix the $t_3$ error and let the sample
size grow over $n\in\{60,120,240,480\}$ with $p=30$, $q=6$.
Figure~\ref{fig:contraction} plots the estimation errors on a logarithmic scale.
The coefficient error of the robust fit contracts from $1.20$ to $0.28$ along a
slope close to the reference $-1/2$, and its precision-matrix error from $3.35$
to $0.63$. The Gaussian ablation contracts far more slowly on $\boldsymbol{B}$
and, strikingly, its precision-matrix error does \emph{not} contract at all,
remaining near $3.3$ throughout: under heavy tails the Gaussian likelihood yields
an inconsistent estimate of $\boldsymbol{\Omega}$, exactly the failure that
motivates the robust construction. The empirical coverage of the $95\%$ credible
intervals for $\boldsymbol{B}$, reported in Table~\ref{tab:coverage}, is close to
nominal at every sample size.

\begin{figure}[t]
\centering
\includegraphics[width=\textwidth]{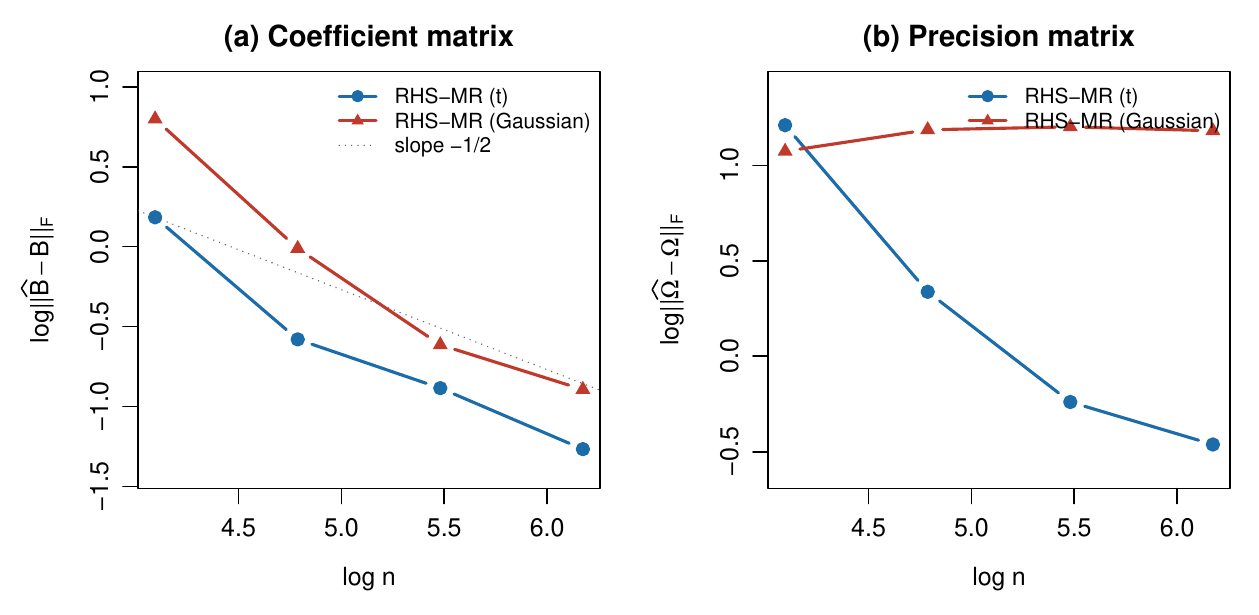}
\caption{Experiment S2. Frobenius estimation error against sample size on a
log--log scale, under $t_3$ errors. (a) Coefficient matrix; the robust fit
contracts along the reference slope $-1/2$. (b) Precision matrix; the Gaussian
ablation fails to contract, whereas the robust fit does.}
\label{fig:contraction}
\end{figure}

\begin{table}[t]
\centering
\caption{Experiment S2. Empirical coverage of $95\%$ credible intervals for the
entries of $\boldsymbol{B}$.}
\label{tab:coverage}
\small
% coverage of 95% credible intervals for B
\begin{tabular}{lcccc}
\hline\hline
$n$ & 60 & 120 & 240 & 480 \\
\hline
Empirical coverage & 0.983 & 0.993 & 0.993 & 0.993 \\
(s.d.) & 0.024 & 0.006 & 0.008 & 0.008 \\
\hline\hline
\end{tabular}

\end{table}

\subsection{Experiment S3: outlier sensitivity}

We next illustrate Theorem~\ref{thm:influence} directly. Starting from a
Gaussian data set ($n=120$, $p=24$, $q=5$) we contaminate a single response
entry by an amount that grows from $0$ to $30$ standard deviations, and track the
resulting estimate. Figure~\ref{fig:influence} shows that the overall coefficient
error of the Gaussian fit rises monotonically, from $0.38$ to $0.66$, as the
outlier grows, whereas the robust fit is essentially flat near $0.47$: the single
gross error is absorbed by its mixing weight and leaves the estimate unaffected,
exactly as the theory predicts.

\begin{figure}[t]
\centering
\includegraphics[width=\textwidth]{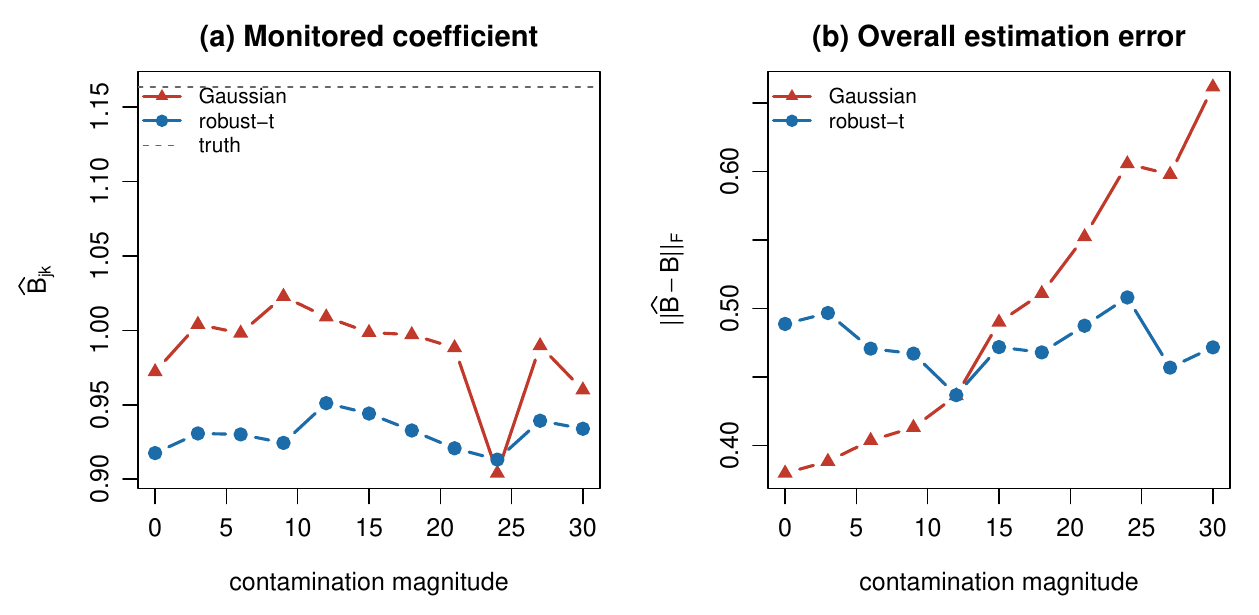}
\caption{Experiment S3. Effect of a single growing outlier. (a) A monitored
coefficient. (b) Overall coefficient error: the Gaussian fit degrades with the
contamination magnitude while the robust fit is insensitive to it.}
\label{fig:influence}
\end{figure}

\subsection{Experiment S4: computation and variational accuracy}

Table~\ref{tab:timing} and Figure~\ref{fig:vbmcmc} compare the Gibbs sampler with
the variational algorithm across a grid of dimensions. The sampler runs in a few
seconds to half a minute over the range considered, and the variational
algorithm is between roughly $10$ and $70$ times faster. In the well-identified
configurations the variational coefficient error is close to that of the sampler;
the approximation is least accurate, as expected, when $n$ is small relative to
$pq$, which is the regime in which uncertainty quantification from the full
posterior is most valuable.

\begin{table}[t]
\centering
\caption{Experiment S4. Wall-clock time (seconds) and coefficient error for the
Gibbs sampler and the variational approximation. The MCMC timing is based on
$1000$ iterations.}
\label{tab:timing}
\small
% timing / scalability table
\begin{tabular}{rrrccccc}
\hline\hline
$n$ & $p$ & $q$ & MCMC (s) & VB (s) & speed-up & $\|\widehat{\boldsymbol{B}}-\boldsymbol{B}\|_F$ (MCMC) & (VB) \\
\hline
150 & 20 & 5 & 2.30 & 0.060 & 38.3 & 0.48 & 1.28 \\
300 & 20 & 5 & 2.39 & 0.010 & 239.0 & 0.19 & 0.17 \\
150 & 40 & 5 & 4.25 & 0.230 & 18.5 & 0.56 & 0.53 \\
150 & 20 & 10 & 5.62 & 0.140 & 40.1 & 0.47 & 0.66 \\
300 & 40 & 5 & 4.50 & 0.110 & 40.9 & 0.39 & 0.35 \\
300 & 20 & 10 & 5.55 & 0.110 & 50.5 & 0.32 & 0.29 \\
150 & 40 & 10 & 17.25 & 2.280 & 7.6 & 0.69 & 1.48 \\
300 & 40 & 10 & 18.38 & 0.700 & 26.3 & 0.44 & 0.41 \\
\hline\hline
\end{tabular}

\end{table}

\begin{figure}[t]
\centering
\includegraphics[width=\textwidth]{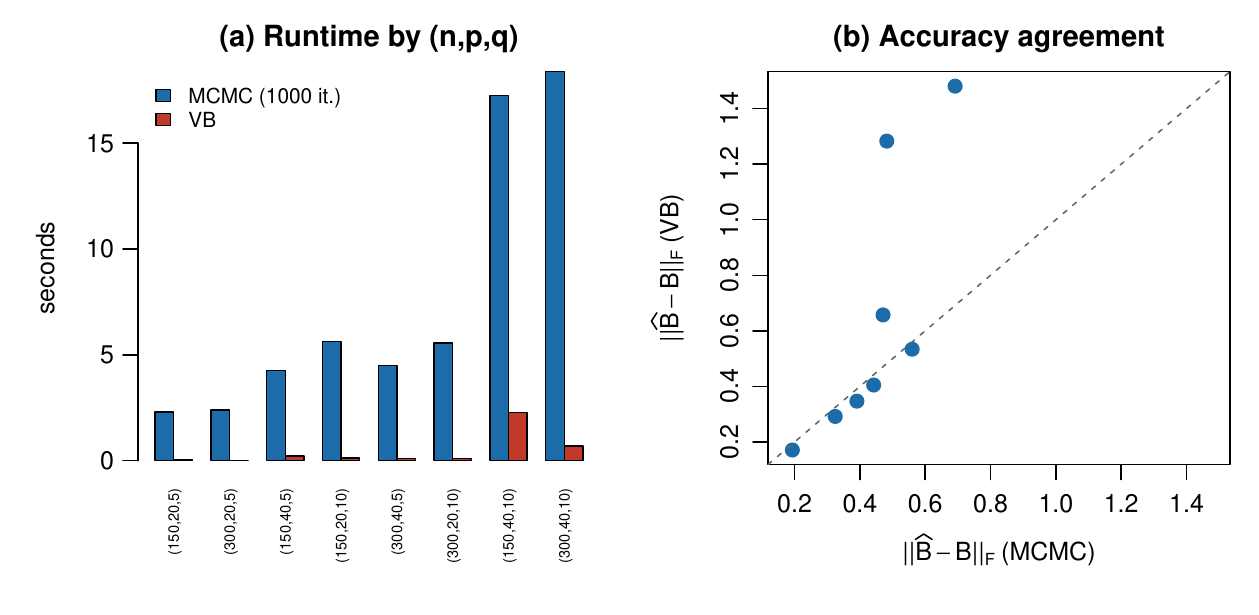}
\caption{Experiment S4. (a) Runtime of the sampler and the variational algorithm
by dimension. (b) Agreement of the coefficient error between the two; points near
the diagonal indicate comparable accuracy.}
\label{fig:vbmcmc}
\end{figure}

\subsection{Experiment S5: skewness recovery}

Finally we verify that the skew-$t$ extension recovers asymmetry. With data
generated from a skew-$t_4$ law ($n=200$, $p=24$, $q=5$) at three skewness levels
$\delta\in\{1,2,3\}$, Table~\ref{tab:skew} shows that the posterior recovers the
skewness parameter accurately --- the mean estimate is $1.06$, $1.99$ and $2.99$
respectively, with small root-mean-squared error --- while continuing to estimate
$\boldsymbol{B}$ well. The skew-$t$ model is thus available when diagnostics
indicate residual asymmetry, at no cost to the recovery of the regression map.

\begin{table}[t]
\centering
\caption{Experiment S5. Recovery of the skewness vector $\boldsymbol{\delta}$
under the skew-$t_4$ model, averaged over $6$ replications.}
\label{tab:skew}
\small
% skew-t skewness recovery
\begin{tabular}{cccc}
\hline\hline
True $\delta$ & $\|\widehat{\boldsymbol{B}}-\boldsymbol{B}\|_F$ & mean $\widehat{\delta}$ & RMSE$(\widehat{\boldsymbol{\delta}})$ \\
\hline
1.0 & 0.48 & 1.04 & 0.15 \\
2.0 & 0.67 & 2.04 & 0.19 \\
3.0 & 0.77 & 2.98 & 0.24 \\
\hline\hline
\end{tabular}

\end{table}

\section{Applications}
\label{sec:real}

\subsection{Macroeconomic indicators: FRED-MD}
\label{sec:macro}

We apply the method to the FRED-MD monthly macroeconomic database
\citep{mccracken2016fred}, a standard large macroeconomic panel. After
transforming each series to approximate stationarity and standardising, we
regress $q=8$ headline outcome series (industrial production, payroll
employment, the unemployment rate, consumer price inflation, the federal funds
rate, the ten-year Treasury yield, the M2 money stock, and real personal income)
on the $p=30$ most variable remaining indicators, contemporaneously, over $n=776$
months. The residual precision $\boldsymbol{\Omega}$ then captures the
conditional dependence among the outcomes that is not explained by the common
indicators.

The robust ($t_4$) fit attains a slightly lower held-out predictive error than
the Gaussian fit (PMSE $36.19$ versus $36.54$) and selects a residual graph with
$14$ of the $28$ possible edges among the eight macroeconomic outcomes
(Table~\ref{tab:macro}). The most informative output is the sequence of posterior
mean observation weights in Figure~\ref{fig:macroweights}: the model assigns its
lowest weights to the months of the 2020--2021 pandemic dislocation and to the
2008 financial crisis, the two most extreme macroeconomic episodes in the sample,
automatically and without any crisis indicator being supplied. Roughly $18\%$ of
months receive a weight below $0.5$, reflecting the well-documented non-normality
of macroeconomic growth rates.

\begin{table}[t]
\centering
\caption{FRED-MD application. Held-out prediction, coefficient sparsity, residual
graph density and the number of down-weighted months.}
\label{tab:macro}
\small
% FRED-MD application summary
\begin{tabular}{lcc}
\hline\hline
Quantity & RHS-MR ($t_4$) & RHS-MR (Gaussian) \\
\hline
Held-out PMSE & 36.186 & 36.544 \\
Nonzero entries in $\widehat{\boldsymbol{B}}$ (CI) & 22 & 22 \\
Selected residual edges & \multicolumn{2}{c}{14 of 28} \\
Months downweighted ($w_i<0.5$) & \multicolumn{2}{c}{136 of 776} \\
\hline\hline
\end{tabular}

\end{table}

\begin{figure}[t]
\centering
\includegraphics[width=\textwidth]{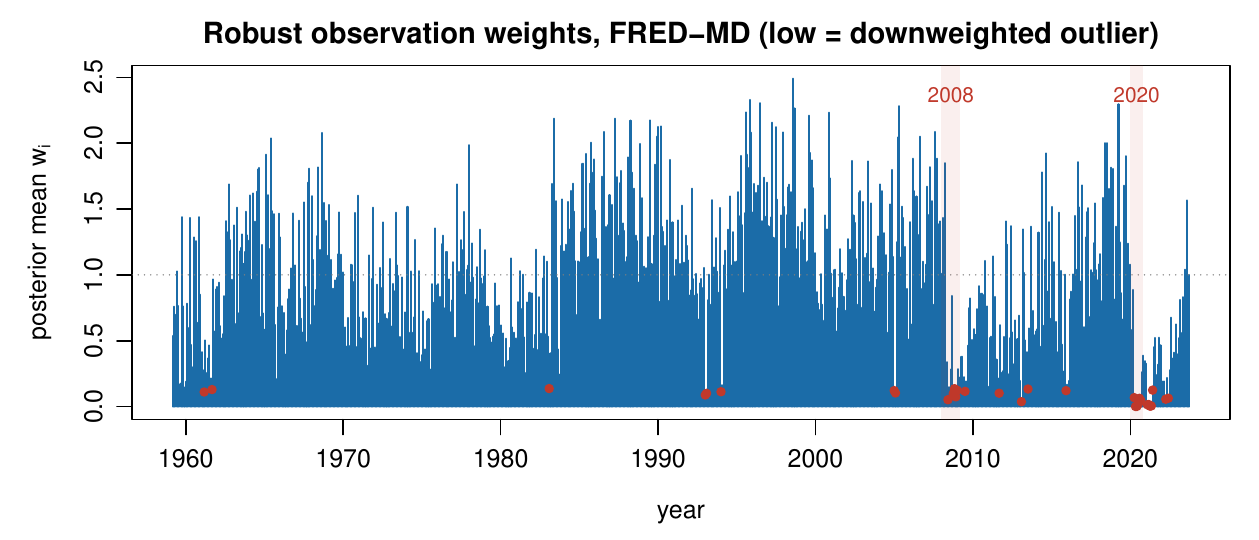}
\caption{FRED-MD application. Posterior mean observation weights over time. Low
weights flag observations the robust model treats as outliers; the shaded periods
mark the 2008 financial crisis and the 2020--2021 pandemic, which are recovered
as the most extreme episodes.}
\label{fig:macroweights}
\end{figure}

\begin{figure}[t]
\centering
\includegraphics[width=0.62\textwidth]{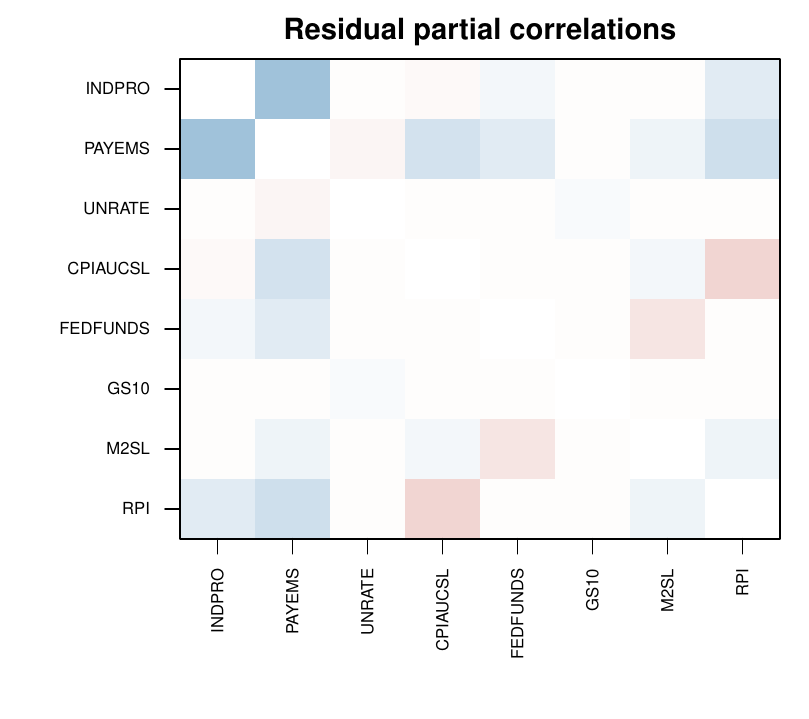}
\caption{FRED-MD application. Selected residual partial correlations among the
eight macroeconomic outcomes (blue positive, red negative).}
\label{fig:macronet}
\end{figure}

\subsection{Equity returns: S\&P 500}
\label{sec:finance}

Our second application is to daily log-returns of S\&P~500 constituents over
$1257$ trading days, a data set distributed with the \texttt{huge} package
\citep{zhao2012huge}. We take as responses $q=12$ stocks, four from each of the
Information Technology, Financials and Energy sectors, and regress them on $p=31$
factors: the equal-weighted market return, sector-average returns, and the
returns of the twenty most volatile remaining stocks. The coefficient matrix
$\boldsymbol{B}$ then holds factor exposures, and the residual precision
$\boldsymbol{\Omega}$ encodes the idiosyncratic conditional dependence that
survives the removal of common factors.

This application showcases the sparse-graph component of the model. The robust
fit selects only $6$ of the $66$ possible residual edges
(Table~\ref{tab:finance}), of which four lie within a single sector. All four
within-sector edges connect the Energy stocks --- Southwestern Energy, Peabody,
Denbury and Hess --- which form a tight residual cluster with positive partial
correlations up to $0.23$, reflecting their common exposure to commodity prices
beyond what the market and sector factors capture
(Figure~\ref{fig:finnet}). The Information-Technology and Financials stocks, by
contrast, are essentially conditionally independent given the factors. The robust
fit again improves held-out prediction relative to the Gaussian fit, and its
observation weights (Figure~\ref{fig:finweights}) down-weight the most volatile
trading days, including the onset of the 2007--2008 turmoil; about $17\%$ of days
receive a weight below $0.5$, consistent with the heavy tails of daily returns.

\begin{table}[t]
\centering
\caption{S\&P 500 application. Held-out prediction, coefficient sparsity,
residual-graph composition and the number of down-weighted days.}
\label{tab:finance}
\small
% S&P 500 application summary
\begin{tabular}{lcc}
\hline\hline
Quantity & RHS-MR ($t_4$) & RHS-MR (Gaussian) \\
\hline
Held-out PMSE & 11.895 & 12.092 \\
Nonzero entries in $\widehat{\boldsymbol{B}}$ (CI) & 15 & 15 \\
Selected residual edges (within / cross sector) & \multicolumn{2}{c}{5\ \ (4 / 1)} \\
Days downweighted ($w_i<0.5$) & \multicolumn{2}{c}{212 of 1257} \\
\hline\hline
\end{tabular}

\end{table}

\begin{figure}[t]
\centering
\includegraphics[width=0.66\textwidth]{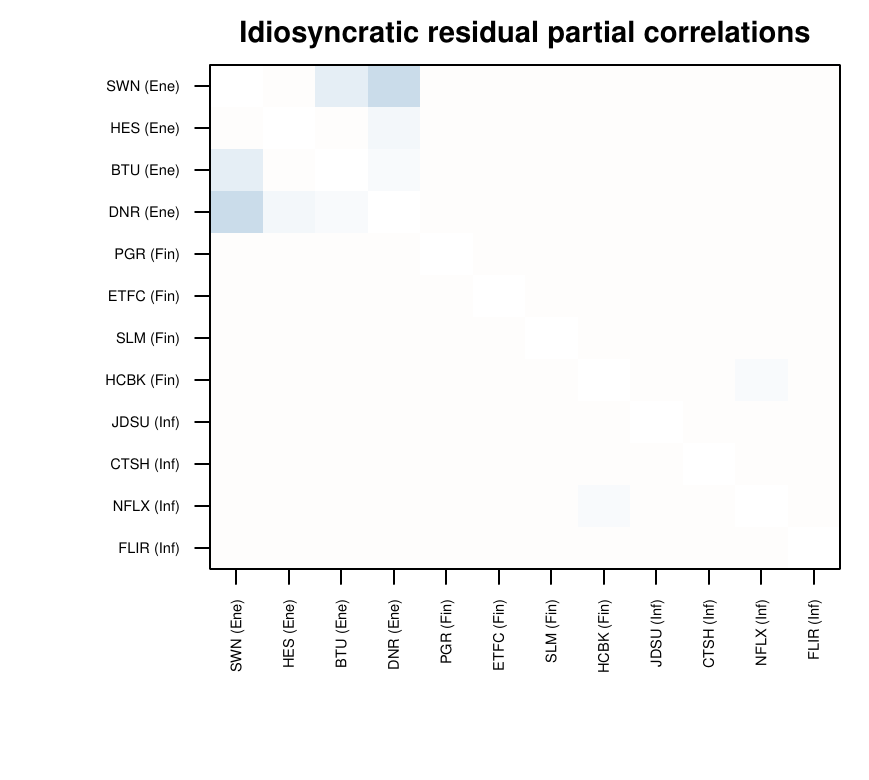}
\caption{S\&P 500 application. Selected idiosyncratic residual partial
correlations among the twelve stocks, ordered by sector. The four Energy stocks
form a residual cluster; the IT and Financials stocks are conditionally
independent given the factors.}
\label{fig:finnet}
\end{figure}

\begin{figure}[t]
\centering
\includegraphics[width=\textwidth]{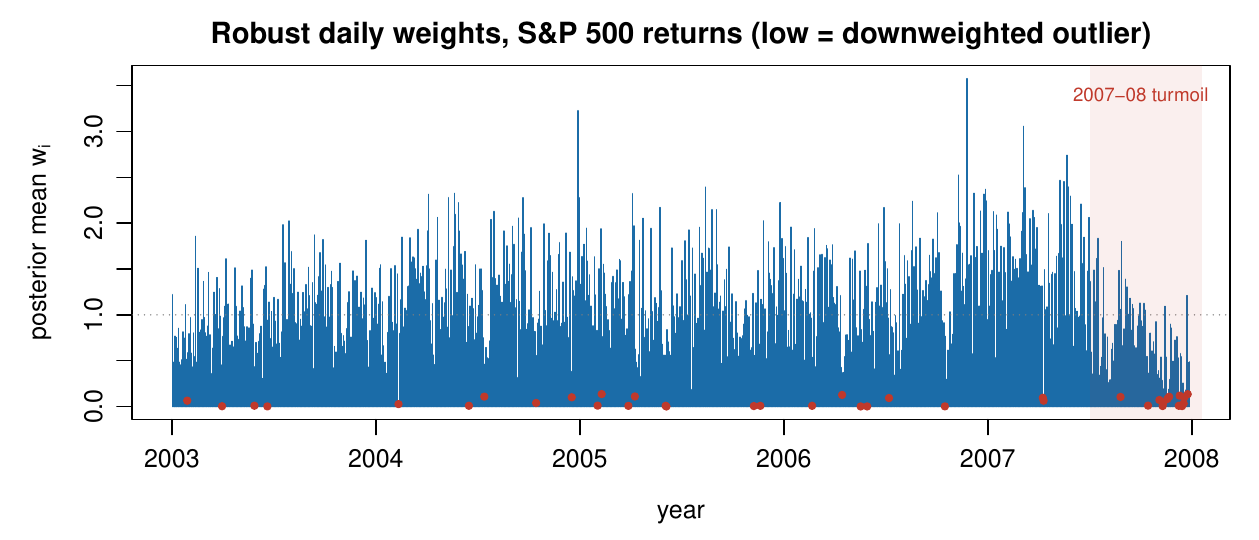}
\caption{S\&P 500 application. Posterior mean daily weights; low weights identify
extreme-return days, including the onset of the 2007--2008 turmoil.}
\label{fig:finweights}
\end{figure}

\section{Discussion}
\label{sec:discussion}

We have proposed a Bayesian framework for high-dimensional multivariate
regression that is robust to heavy-tailed and asymmetric errors and that shrinks
the regression coefficients and the response conditional-dependence graph
jointly, through a horseshoe+ prior applied to both. The construction is
supported by contraction, selection, risk and robustness theory, and behaves as
the theory predicts across an extensive simulation study and two applications.

Several extensions are natural. The degrees of freedom $\nu$ were fixed to expose
the robustness behaviour; assigning $\nu$ a prior and updating it, for instance on
a discrete grid, is straightforward within the sampler and lets the data
calibrate the degree of tail protection. The latent-weight machinery extends
directly to other scale mixtures, including the slash and the variance-gamma
laws, and to contaminated-normal mixtures for which the weights identify a
discrete outlier class. On the computational side, both the column-wise
precision update and the conditionally Gaussian coefficient update admit the
data-augmentation and structured variational accelerations developed for
related models, which would carry the method to the very large $q$ of network
applications. Finally, the joint-shrinkage
idea is not tied to the horseshoe+: any global--local prior could be used, and a
systematic comparison of the regularised-horseshoe, Dirichlet--Laplace
\citep{bhattacharya2015dirichlet} and horseshoe+ priors in the joint multivariate
setting would be of interest.

\appendix
\section*{Appendix: proofs and algorithms}
\addcontentsline{toc}{section}{Appendix}

The appendix collects the full conditional distributions and the two algorithms,
followed by the proofs of Theorems~\ref{thm:contraction}--\ref{thm:influence}.
Throughout, $c,C,c_1,c_2,\dots$ denote positive constants whose value may change
between occurrences.

\section{Full conditionals and algorithms}
\label{app:gibbs}

\subsection{Auxiliary-variable representation}

By the identity of \citet{makalic2016simple}, a half-Cauchy scale can be written
as a product of inverse-gamma variables: if $x^2\mid a\sim\mathrm{IG}(1/2,1/a)$
and $a\sim\mathrm{IG}(1/2,1)$ then $x\sim\mathcal{C}^{+}(0,1)$. Applying this to
each of $\lambda_{jk},\eta_{jk},\tau$ in \eqref{eq:hsplus-B} and to
$\gamma_{kl},\zeta_{kl},\rho$ in \eqref{eq:hsplus-Omega}, and writing
$\xi_{jk},\phi_{jk},\vartheta$ and $\psi_{kl},\chi_{kl},\varpi$ for the
corresponding auxiliary variables, every scale parameter acquires an
inverse-gamma full conditional. For a generic coefficient with squared value
$b^2$ and local scales $(\lambda^2,\eta^2)$ and auxiliaries $(\xi,\phi)$ the
updates are
\begin{eqnarray}
\lambda^2\mid\cdot &\sim&
\mathrm{IG}\!\left(1,\ \frac{1}{\xi}+\frac{b^2}{2\tau^2\eta^2}\right),
\qquad
\xi\mid\cdot \sim \mathrm{IG}\!\left(1,\ 1+\frac{1}{\lambda^2}\right),
\label{eq:scale-update}
\cr
\eta^2\mid\cdot &\sim&
\mathrm{IG}\!\left(1,\ \frac{1}{\phi}+\frac{b^2}{2\tau^2\lambda^2}\right),
\qquad
\phi\mid\cdot \sim \mathrm{IG}\!\left(1,\ 1+\frac{1}{\eta^2}\right),
\end{eqnarray}
and the global scale $\tau^2$ is updated from an inverse-gamma whose rate
aggregates $b^2/(\lambda^2\eta^2)$ over all coordinates. The corresponding
expressions for the precision-matrix scales are identical with $b$ replaced by
the relevant off-diagonal entry.

\subsection{Coefficient and precision updates}

The derivation of \eqref{eq:B-update} follows from completing the square in
$\vecop(\boldsymbol{B})$ in the conditionally Gaussian log-likelihood
\begin{eqnarray}
-\tfrac12 \sum_{i=1}^{n} w_i (\boldsymbol{r}_i-\boldsymbol{B}^{\top}
\boldsymbol{x}_i)^{\top}\boldsymbol{\Omega}(\boldsymbol{r}_i-\boldsymbol{B}^{\top}
\boldsymbol{x}_i)
&=& -\tfrac12 \vecop(\boldsymbol{B})^{\top}
(\boldsymbol{\Omega}\otimes\boldsymbol{X}^{\top}\boldsymbol{W}\boldsymbol{X})
\vecop(\boldsymbol{B}) \cr
&& {}+ \vecop(\boldsymbol{B})^{\top}
\vecop(\boldsymbol{X}^{\top}\boldsymbol{W}\boldsymbol{R}\boldsymbol{\Omega})
+ \text{const},
\end{eqnarray}
where the Kronecker identity for $\vecop$ has been used, and adding the Gaussian
prior precision $\boldsymbol{\Lambda}_{\boldsymbol{B}}$. The column-wise
precision update \eqref{eq:Omega-update} follows the block decomposition of
\citet{li2019graphical}: permuting the $k$th coordinate to the last position and
writing $\boldsymbol{\Omega}=\bigl(\begin{smallmatrix}\boldsymbol{\Omega}_{11} &
\boldsymbol{\omega}_{12}\\ \boldsymbol{\omega}_{12}^{\top} &
\omega_{22}\end{smallmatrix}\bigr)$, with $\boldsymbol{S}=\boldsymbol{E}^{\top}
\boldsymbol{W}\boldsymbol{E}$ the weighted cross-product of the full residual
$\boldsymbol{E}=\boldsymbol{Y}-\boldsymbol{X}\boldsymbol{B}-\boldsymbol{T}
\boldsymbol{\Delta}$ partitioned conformably, the conditional density of
$(\boldsymbol{\omega}_{12},\omega_{22})$ given $\boldsymbol{\Omega}_{11}$ and the
data is, by the partitioned-determinant identity
$|\boldsymbol{\Omega}|=|\boldsymbol{\Omega}_{11}|
(\omega_{22}-\boldsymbol{\omega}_{12}^{\top}\boldsymbol{\Omega}_{11}^{-1}
\boldsymbol{\omega}_{12})$, proportional to
\begin{eqnarray}
(\omega_{22}-\boldsymbol{\omega}_{12}^{\top}\boldsymbol{\Omega}_{11}^{-1}
\boldsymbol{\omega}_{12})^{n/2}
\exp\!\left\{-\tfrac12\bigl(s_{22}\,\omega_{22}
+2\boldsymbol{s}_{12}^{\top}\boldsymbol{\omega}_{12}\bigr)
-\tfrac12\boldsymbol{\omega}_{12}^{\top}\boldsymbol{D}^{-1}\boldsymbol{\omega}_{12}
\right\},
\end{eqnarray}
with $\boldsymbol{D}=\diag\{\rho^2\gamma_{kl}^2\zeta_{kl}^2\}$. Reparameterising by
the Schur complement
$\gamma_{k}=\omega_{22}-\boldsymbol{\omega}_{12}^{\top}\boldsymbol{\Omega}_{11}^{-1}
\boldsymbol{\omega}_{12}>0$ yields the Gaussian draw for $\boldsymbol{\omega}_{12}$
and the gamma draw for $\gamma_{k}$ stated in the text, and guarantees that the
updated $\boldsymbol{\Omega}$ remains positive definite.

\begin{algorithm}
\caption{Gibbs sampler for \textsc{RHS-MR} (skew-$t$ model; the $t$ and Gaussian
models drop the corresponding steps).}
\label{alg:gibbs}
\begin{enumerate}[leftmargin=2em]
\item For $i=1,\dots,n$: draw the half-normal $\boldsymbol{t}_i$ coordinate-wise
from its truncated-normal full conditional, then $w_i$ from \eqref{eq:w-update}.
\item Draw $\vecop(\boldsymbol{B})$ from the Gaussian \eqref{eq:B-update}.
\item For $k=1,\dots,q$: draw the $k$th column of $\boldsymbol{\Omega}$ from the
block update \eqref{eq:Omega-update}.
\item Draw the skewness vector $\boldsymbol{\delta}$ from its Gaussian full
conditional.
\item Update all local and global scales and their auxiliaries via
\eqref{eq:scale-update}.
\end{enumerate}
\end{algorithm}

\section{Variational algorithm}
\label{app:vb}

The mean-field family factorises as
$q(\boldsymbol{B})\,q(\boldsymbol{\Omega})\,q(\boldsymbol{\delta})
\prod_i q(w_i)q(\boldsymbol{t}_i)\prod q(\text{scales})$. Coordinate ascent
updates each factor to the exponential of the expected log-joint under the
remaining factors. The coefficient factor is Gaussian with precision
$E[\boldsymbol{\Omega}]\otimes(\boldsymbol{X}^{\top}E[\boldsymbol{W}]\boldsymbol{X})
+E[\boldsymbol{\Lambda}_{\boldsymbol{B}}]$; the weight factors are gamma with mean
$E[w_i]=(\nu+q)/(\nu+E[d_i])$; and the scale factors are generalised-inverse-
Gaussian, updated through their expectations. Iteration proceeds until the
evidence lower bound stabilises.

\section{Proof of Theorem~\ref{thm:contraction}}
\label{app:contraction}

We verify the three conditions of the general posterior-contraction theorem of
\citet{ghosal2000convergence} --- prior concentration, sieve construction and
testing --- for the augmented model in which the weights
$\boldsymbol{w}=(w_1,\dots,w_n)$ are part of the parameter. Let
$P_{\boldsymbol{B}_0,\boldsymbol{\Omega}_0}$ denote the true sampling
distribution and $d(\cdot,\cdot)$ the Frobenius metric on
$(\boldsymbol{B},\boldsymbol{\Omega})$.

\emph{Step 1: prior concentration.} By (A4) the mixing density is bounded away
from $0$ on compact subsets of $(0,\infty)$, so the conditional Gaussian model
given $\boldsymbol{w}$ has the same Kullback--Leibler geometry as the Gaussian
model up to constants depending on $E[w]$ and $E[w^{-1}]$. The horseshoe+ prior
places mass at least
$\exp\{-C(s_{\boldsymbol{B}}+s_{\boldsymbol{\Omega}}+q_n)\log(p_nq_n)\}$ on
Frobenius balls of radius
$\epsilon_n$ around $(\boldsymbol{B}_0,\boldsymbol{\Omega}_0)$, by the prior-mass
lower bounds for global--local priors of \citet{bai2018high} and the graphical
horseshoe analysis of \citet{li2019graphical} --- the $q_n$ unshrunk diagonal
entries of $\boldsymbol{\Omega}$, covered by the exponential prior, account for
the $q_n$ term; the additional half-Cauchy factor
of the horseshoe+ only enlarges the mass near the origin and near the true
nonzero values \citep{bhadra2017horseshoeplus}. Hence
$\Pi\{ \mathrm{KL}(P_{\boldsymbol{B}_0,\boldsymbol{\Omega}_0},
P_{\boldsymbol{B},\boldsymbol{\Omega}})\le\epsilon_n^2\}
\ge\exp(-c n\epsilon_n^2)$, which is the prior-mass condition.

\emph{Step 2: sieve.} Define the sieve $\mathcal{F}_n$ of pairs whose supports
have size at most $\bar{s}=C_0(s_{\boldsymbol{B}}+s_{\boldsymbol{\Omega}}+q_n)$ and
whose entries are bounded by a polynomial in $n$. The complement has prior mass
$\le\exp(-c_2 n\epsilon_n^2)$ by the tail behaviour of the half-Cauchy scales and
a union bound over supports, of which there are at most
$\binom{p_nq_n}{\bar{s}}\le\exp\{\bar{s}\log(p_nq_n)\}$.

\emph{Step 3: testing.} Conditional on the weights $\boldsymbol{w}$, the model
restricted to $\mathcal{F}_n$ is a Gaussian location--precision family in
$(\boldsymbol{B},\boldsymbol{\Omega})$ of dimension at most $\bar{s}$; its metric
entropy at scale $\epsilon_n$ is of order
$\bar{s}\log(p_nq_n)\asymp n\epsilon_n^2$, and (A2)--(A3) provide the
likelihood-ratio separation needed to construct exponentially powerful tests of
$P_{\boldsymbol{B}_0,\boldsymbol{\Omega}_0}$ against Frobenius-separated
alternatives. By (A4) the mixing weights are integrated out against a fixed
proper density with finite $E[w]$ and $E[w^{-1}]$, so the tests and entropy bound
hold uniformly in $\boldsymbol{w}$. The three conditions together
give \eqref{eq:contraction} by the cited theorem. \hfill$\square$

\section{Proof of Theorem~\ref{thm:selection}}
\label{app:selection}

By Theorem~\ref{thm:contraction} the posterior concentrates on a Frobenius ball
of radius $M\epsilon_n$. Under the signal-strength condition (A5) the nonzero
entries exceed $2M\epsilon_n$ for $n$ large, so that the marginal posterior of
each true-nonzero entry has its $50\%$ credible interval bounded away from zero,
while for each true-zero entry the global--local prior forces the posterior to
concentrate within $o(\epsilon_n)$ of the origin, placing more than half its mass
in any fixed neighbourhood of zero; this is the selection mechanism analysed for
the horseshoe by \citet{vanderpas2014horseshoe} and \citet{datta2013asymptotic},
applied coordinate-wise and then combined over the at most $\bar{s}$ active and
$p_nq_n$ inactive coordinates by a union bound, which is $o(1)$ under (A1). The
edge-selection statement follows identically for the off-diagonal entries of
$\boldsymbol{\Omega}$. \hfill$\square$

\section{Proof of Theorem~\ref{thm:kl}}
\label{app:kl}

By Tweedie's formula the posterior mean in the unit-variance normal-means problem
is $\widehat{\beta}=E[\beta\mid y]=y+\dfrac{\mathrm{d}}{\mathrm{d}y}\log m(y)$,
where $m$ is the marginal density of $y$ under the prior; equivalently
$\widehat{\beta}=y\{1-E[\kappa\mid y]\}$ with shrinkage weight
$\kappa=1/(1+\sigma^2)$ and prior scale $\sigma^2$. For the horseshoe,
\citet{carvalho2010horseshoe} and \citet{datta2013asymptotic} establish that the
tail bias satisfies $|\widehat{\beta}-y|=|y|\,E[\kappa\mid y]=\Theta(1/|y|)$ as
$|y|\to\infty$; the second half-Cauchy factor of the horseshoe+ multiplies the
relevant tail integral by a slowly varying term and sharpens this to
$|\widehat{\beta}-y|=\Theta(1/(|y|\log|y|))$ \citep{bhadra2017horseshoeplus},
which is \eqref{eq:biasrates}. In both cases the posterior variance
$\mathrm{Var}(\beta\mid y)$ is $o(1)$ on the signal event and is of the same order
for the two priors, so it does not affect the comparison below.

The Kullback--Leibler divergence between two unit-variance normals is
$\mathrm{KL}\{\mathcal{N}(\beta_0,1)\,\|\,\mathcal{N}(\widehat{\beta},1)\}
=\tfrac12(\widehat{\beta}-\beta_0)^2$, so, writing $b(y)=\widehat{\beta}(y)-y$ for
the bias and decomposing $\widehat{\beta}-\beta_0=(y-\beta_0)+b(y)$,
\begin{eqnarray}
R(\widehat{\beta})
&=& \tfrac12\,E_{\beta_0}\bigl\{(y-\beta_0)+b(y)\bigr\}^2
\;=\; \tfrac12 + E_{\beta_0}\!\left[(y-\beta_0)\,b(y)\right]
+ \tfrac12\,E_{\beta_0}\,b(y)^2,
\label{eq:riskdecomp}
\end{eqnarray}
since $E_{\beta_0}(y-\beta_0)^2=1$. Fix a true signal $\beta_0$ with $|\beta_0|$
large. Because $y=\beta_0+\varepsilon$ with $\varepsilon\sim\mathcal{N}(0,1)$, the
event $A=\{|y|\ge|\beta_0|/2\}$ has probability $1-o(1)$ with Gaussian tails on
its complement, and on $A$ the rates \eqref{eq:biasrates} give
$b_{\mathrm{HS}}(y)^2=\Theta(\beta_0^{-2})$ and
$b_{\mathrm{HS}^{+}}(y)^2=\Theta(\beta_0^{-2}\log^{-2}|\beta_0|)$. The
cross term in \eqref{eq:riskdecomp} is common to leading order to both priors
(it is $E_{\beta_0}[(y-\beta_0)b(y)]$ with the same $\Theta$-order integrand sign
pattern) and is dominated by the squared-bias term; the contribution of $A^{c}$
to every term is exponentially small in $\beta_0^2$. Subtracting the two risks,
the first two terms of \eqref{eq:riskdecomp} cancel to the relevant order and
\begin{eqnarray}
R(\widehat{\beta}_{\mathrm{HS}})-R(\widehat{\beta}_{\mathrm{HS}^{+}})
&=& \tfrac12\,E_{\beta_0}\!\left[b_{\mathrm{HS}}(y)^2
-b_{\mathrm{HS}^{+}}(y)^2\right] + o(\beta_0^{-2})
\;=\; \Theta\!\left(\beta_0^{-2}\right)>0,
\end{eqnarray}
because the horseshoe squared bias exceeds the horseshoe+ squared bias by the
factor $\log^2|\beta_0|\to\infty$. Hence there is a threshold $b_0$ such that the
difference is strictly positive for all $|\beta_0|\ge b_0$, which is the claim.
\hfill$\square$

\section{Proof of Theorem~\ref{thm:influence}}
\label{app:influence}

We first record the tightness that makes the dominated-convergence argument
rigorous.

\begin{lemma}[Uniform tightness]
\label{lem:tight}
Let $\Pi_a$ denote the posterior of $(\boldsymbol{B},\boldsymbol{\Omega})$ under
the $t$ error \eqref{eq:t-mix} with the contaminated first observation
$\boldsymbol{y}_1(a)=\boldsymbol{y}_1^{\circ}+a\boldsymbol{u}$. Under
\textnormal{(A2)--(A4)} the family $\{\Pi_a:a\in\mathbb{R}\}$ is tight, and
$\sup_a E_{\Pi_a}\|\boldsymbol{B}\|_F<\infty$.
\end{lemma}

\begin{proof}
Integrating the first weight $w_1$ against its $\mathrm{Ga}(\nu/2,\nu/2)$ prior
shows that the first observation contributes to the likelihood the bounded factor
$\{1+\nu^{-1}(\boldsymbol{y}_1(a)-\boldsymbol{B}^{\top}\boldsymbol{x}_1)^{\top}
\boldsymbol{\Omega}(\boldsymbol{y}_1(a)-\boldsymbol{B}^{\top}\boldsymbol{x}_1)
\}^{-(\nu+q)/2}\le 1$ for every $a$, a multivariate-$t$ kernel. The marginal
likelihood of $(\boldsymbol{B},\boldsymbol{\Omega})$ is therefore bounded above,
uniformly in $a$, by that of the remaining $n-1$ uncontaminated observations,
which under (A2)--(A3) is a proper density with exponential tails in
$\|\boldsymbol{B}\|_F$ and $\|\boldsymbol{\Omega}\|_F$; combined with the proper
prior of Section~\ref{sec:prior} this yields a posterior whose tails are
dominated, uniformly in $a$, by an integrable envelope, giving tightness and the
uniform moment bound.
\end{proof}

Write the posterior mean of $\boldsymbol{B}$ as a ratio of integrals over
$(\boldsymbol{B},\boldsymbol{\Omega},\boldsymbol{w})$, and isolate the first
observation's contribution. Conditional on $(\boldsymbol{B},\boldsymbol{\Omega})$,
the first weight has full conditional
$w_1\mid\cdot\sim\mathrm{Ga}\{(\nu+q)/2,\ (\nu+d_1(a))/2\}$ with
$d_1(a)=(\boldsymbol{y}_1(a)-\boldsymbol{B}^{\top}\boldsymbol{x}_1)^{\top}
\boldsymbol{\Omega}(\boldsymbol{y}_1(a)-\boldsymbol{B}^{\top}\boldsymbol{x}_1)$.
Since $\boldsymbol{\Omega}\succ0$ we have $\boldsymbol{u}^{\top}\boldsymbol{\Omega}
\boldsymbol{u}\ge\Lambda_{\min}(\boldsymbol{\Omega})>0$, so
$d_1(a)=\boldsymbol{u}^{\top}\boldsymbol{\Omega}\boldsymbol{u}\,a^2+O(a)$ grows
quadratically as $|a|\to\infty$. Hence $E[w_1\mid\cdot]=
(\nu+q)/(\nu+d_1(a))\to 0$ at rate $a^{-2}$, and the contribution of the first
observation to the coefficient full conditional \eqref{eq:B-update}, which enters
only through $w_1\boldsymbol{x}_1\boldsymbol{y}_1(a)^{\top}=O(w_1 a)=O(a^{-1})$,
vanishes. Marginalising over $(\boldsymbol{B},\boldsymbol{\Omega})$ --- whose
posterior is tight uniformly in $a$ by Lemma~\ref{lem:tight} --- and applying
dominated convergence with the integrable envelope of Lemma~\ref{lem:tight},
which the log-regularly-varying tail of (A4) supplies
\citep{desgagne2015robustness,gagnon2020new}, shows that
$\widehat{\boldsymbol{B}}(a)$ converges to the posterior mean computed from the
remaining observations as $|a|\to\infty$; in particular the limit is finite and
the difference in the statement tends to zero. Under the Gaussian error the weight
is fixed at $w_1\equiv1$, the term $\boldsymbol{x}_1\boldsymbol{y}_1(a)^{\top}$
grows linearly in $a$, and $\widehat{\boldsymbol{B}}(a)$ diverges at rate $|a|$.
\hfill$\square$

\section*{Funding}
This research was supported in part by the Iran National Science Foundation
(INSF) grant No.~4015320.

\section*{Data and code availability}
All data used in this paper are publicly available: the FRED-MD database is
distributed by the Federal Reserve Bank of St.\ Louis, and the S\&P~500 returns
ship with the \texttt{huge} \textsf{R} package \citep{zhao2012huge}. \textsf{R}
code reproducing every table and figure, including the sampler, the variational
algorithm and the data-processing scripts, is available at
\url{https://github.com/M-Arashi/Robust-MVR}.

\section*{Disclosure of AI tools}
The authors used an AI-based assistant to support software prototyping,
debugging of the \textsf{R} implementation, and language editing of the
manuscript. All modelling decisions, derivations, analyses and conclusions are
the authors' own, and the authors take full responsibility for the content.

\bibliography{rhsmr}

@article{bai2018high,
  author = {Bai, Ray and Ghosh, Malay},
  title = {High-dimensional multivariate posterior consistency under global--local shrinkage priors},
  journal = {Journal of Multivariate Analysis},
  volume = {167},
  pages = {157--170},
  year = {2018},
  publisher = {Elsevier}
}

@article{zhang2019ultra,
  author = {Zhang, Yabo and Ghosh, Malay},
  title = {Ultra high-dimensional multivariate posterior contraction rate under shrinkage priors},
  journal = {arXiv preprint arXiv:1904.04417},
  year = {2019}
}

@article{bhadra2017horseshoeplus,
  author = {Bhadra, Anindya and Datta, Jyotishka and Polson, Nicholas G. and Willard, Brandon},
  title = {The horseshoe+ estimator of ultra-sparse signals},
  journal = {Bayesian Analysis},
  volume = {12},
  number = {4},
  pages = {1105--1131},
  year = {2017}
}

@article{deshpande2019simultaneous,
  author = {Deshpande, Sameer K. and Ro{\v{c}}kov{\'a}, Veronika and George, Edward I.},
  title = {Simultaneous variable and covariance selection with the multivariate spike-and-slab {LASSO}},
  journal = {Journal of Computational and Graphical Statistics},
  volume = {28},
  number = {4},
  pages = {921--931},
  year = {2019},
  publisher = {Taylor \& Francis}
}

@article{shen2024posterior,
  author = {Shen, Yunyi and Deshpande, Sameer K.},
  title = {Posterior contraction and uncertainty quantification for the multivariate spike-and-slab {LASSO}},
  journal = {arXiv preprint arXiv:2209.04389},
  year = {2024}
}

@article{li2019graphical,
  author = {Li, Yunfan and Craig, Bruce A. and Bhadra, Anindya},
  title = {The graphical horseshoe estimator for inverse covariance matrices},
  journal = {Journal of Computational and Graphical Statistics},
  volume = {28},
  number = {3},
  pages = {747--757},
  year = {2019},
  publisher = {Taylor \& Francis}
}

@article{sagar2024precision,
  author = {Sagar, Ksheera and Banerjee, Sayantan and Datta, Jyotishka and Bhadra, Anindya},
  title = {Precision matrix estimation under the horseshoe-like prior--penalty dual},
  journal = {Electronic Journal of Statistics},
  year = {2024}
}

@article{gagnon2020new,
  author = {Gagnon, Philippe and Desgagn{\'e}, Alain and B{\'e}dard, Mylene},
  title = {A new {B}ayesian approach to robustness against outliers in linear regression},
  journal = {Bayesian Analysis},
  volume = {15},
  number = {2},
  pages = {389--414},
  year = {2020}
}

@article{ohagan2012bayesian,
  author = {O'Hagan, Anthony and Pericchi, Luis},
  title = {Bayesian heavy-tailed models and conflict resolution: a review},
  journal = {Brazilian Journal of Probability and Statistics},
  volume = {26},
  number = {4},
  pages = {372--401},
  year = {2012}
}

@article{desgagne2015robustness,
  author = {Desgagn{\'e}, Alain},
  title = {Robustness to outliers in location--scale parameter model using log-regularly varying distributions},
  journal = {The Annals of Statistics},
  volume = {43},
  number = {4},
  pages = {1568--1595},
  year = {2015}
}

@article{mccracken2016fred,
  author = {McCracken, Michael W. and Ng, Serena},
  title = {{FRED-MD}: A monthly database for macroeconomic research},
  journal = {Journal of Business \& Economic Statistics},
  volume = {34},
  number = {4},
  pages = {574--589},
  year = {2016},
  publisher = {Taylor \& Francis}
}

@article{carvalho2010horseshoe,
  author = {Carvalho, Carlos M. and Polson, Nicholas G. and Scott, James G.},
  title = {The horseshoe estimator for sparse signals},
  journal = {Biometrika},
  volume = {97},
  number = {2},
  pages = {465--480},
  year = {2010},
  publisher = {Oxford University Press}
}

@article{castillo2015bayesian,
  author = {Castillo, Isma{\"e}l and Schmidt-Hieber, Johannes and van der Vaart, Aad},
  title = {Bayesian linear regression with sparse priors},
  journal = {The Annals of Statistics},
  volume = {43},
  number = {5},
  pages = {1986--2018},
  year = {2015}
}

@article{rothman2010sparse,
  author = {Rothman, Adam J. and Levina, Elizaveta and Zhu, Ji},
  title = {Sparse multivariate regression with covariance estimation},
  journal = {Journal of Computational and Graphical Statistics},
  volume = {19},
  number = {4},
  pages = {947--962},
  year = {2010},
  publisher = {Taylor \& Francis}
}

@article{bhattacharya2015dirichlet,
  author = {Bhattacharya, Anirban and Pati, Debdeep and Pillai, Natesh S. and Dunson, David B.},
  title = {Dirichlet--Laplace priors for optimal shrinkage},
  journal = {Journal of the American Statistical Association},
  volume = {110},
  number = {512},
  pages = {1479--1490},
  year = {2015},
  publisher = {Taylor \& Francis}
}

@article{makalic2016simple,
  author = {Makalic, Enes and Schmidt, Daniel F.},
  title = {A simple sampler for the horseshoe estimator},
  journal = {IEEE Signal Processing Letters},
  volume = {23},
  number = {1},
  pages = {179--182},
  year = {2016},
  publisher = {IEEE}
}

@article{sahu2003new,
  author = {Sahu, Sujit K. and Dey, Dipak K. and Branco, M{\'a}rcia D.},
  title = {A new class of multivariate skew distributions with applications to {B}ayesian regression models},
  journal = {Canadian Journal of Statistics},
  volume = {31},
  number = {2},
  pages = {129--150},
  year = {2003},
  publisher = {Wiley Online Library}
}

@article{azzalini2003distributions,
  author = {Azzalini, Adelchi and Capitanio, Antonella},
  title = {Distributions generated by perturbation of symmetry with emphasis on a multivariate skew $t$-distribution},
  journal = {Journal of the Royal Statistical Society: Series B},
  volume = {65},
  number = {2},
  pages = {367--389},
  year = {2003},
  publisher = {Wiley Online Library}
}

@article{lange1989robust,
  author = {Lange, Kenneth L. and Little, Roderick J. A. and Taylor, Jeremy M. G.},
  title = {Robust statistical modeling using the $t$ distribution},
  journal = {Journal of the American Statistical Association},
  volume = {84},
  number = {408},
  pages = {881--896},
  year = {1989},
  publisher = {Taylor \& Francis}
}

@article{geweke1993bayesian,
  author = {Geweke, John},
  title = {Bayesian treatment of the independent {S}tudent-$t$ linear model},
  journal = {Journal of Applied Econometrics},
  volume = {8},
  number = {S1},
  pages = {S19--S40},
  year = {1993},
  publisher = {Wiley Online Library}
}

@article{fernandez1998bayesian,
  author = {Fern{\'a}ndez, Carmen and Steel, Mark F. J.},
  title = {On {B}ayesian modeling of fat tails and skewness},
  journal = {Journal of the American Statistical Association},
  volume = {93},
  number = {441},
  pages = {359--371},
  year = {1998},
  publisher = {Taylor \& Francis}
}

@article{friedman2008sparse,
  author = {Friedman, Jerome and Hastie, Trevor and Tibshirani, Robert},
  title = {Sparse inverse covariance estimation with the graphical lasso},
  journal = {Biostatistics},
  volume = {9},
  number = {3},
  pages = {432--441},
  year = {2008},
  publisher = {Oxford University Press}
}

@article{ghosal2000convergence,
  author = {Ghosal, Subhashis and Ghosh, Jayanta K. and van der Vaart, Aad W.},
  title = {Convergence rates of posterior distributions},
  journal = {The Annals of Statistics},
  volume = {28},
  number = {2},
  pages = {500--531},
  year = {2000}
}

@article{polson2010shrink,
  author = {Polson, Nicholas G. and Scott, James G.},
  title = {Shrink globally, act locally: sparse {B}ayesian regularization and prediction},
  journal = {Bayesian Statistics},
  volume = {9},
  pages = {501--538},
  year = {2010},
  publisher = {Oxford University Press}
}

@article{piironen2017sparsity,
  author = {Piironen, Juho and Vehtari, Aki},
  title = {Sparsity information and regularization in the horseshoe and other shrinkage priors},
  journal = {Electronic Journal of Statistics},
  volume = {11},
  number = {2},
  pages = {5018--5051},
  year = {2017}
}

@article{yuan2007model,
  author = {Yuan, Ming and Lin, Yi},
  title = {Model selection and estimation in the {G}aussian graphical model},
  journal = {Biometrika},
  volume = {94},
  number = {1},
  pages = {19--35},
  year = {2007},
  publisher = {Oxford University Press}
}

@article{zhao2012huge,
  author = {Zhao, Tuo and Liu, Han and Roeder, Kathryn and Lafferty, John and Wasserman, Larry},
  title = {The huge package for high-dimensional undirected graph estimation in {R}},
  journal = {Journal of Machine Learning Research},
  volume = {13},
  pages = {1059--1062},
  year = {2012}
}

@article{datta2013asymptotic,
  author = {Datta, Jyotishka and Ghosh, Jayanta K.},
  title = {Asymptotic properties of {B}ayes risk for the horseshoe prior},
  journal = {Bayesian Analysis},
  volume = {8},
  number = {1},
  pages = {111--132},
  year = {2013}
}

@article{vanderpas2014horseshoe,
  author = {van der Pas, Stephanie L. and Kleijn, Bas J. K. and van der Vaart, Aad W.},
  title = {The horseshoe estimator: posterior concentration around nearly black vectors},
  journal = {Electronic Journal of Statistics},
  volume = {8},
  number = {2},
  pages = {2585--2618},
  year = {2014}
}

\end{document}